\documentclass[11pt, letterpaper]{article}
\title{Towards a Proof of the Fourier--Entropy Conjecture?}
\author{
Esty Kelman\thanks{School of Computer Science, Tel Aviv University.} \and
Guy Kindler\thanks{Einstein Institute of Mathematics, Hebrew University of Jerusalem.} \and
Noam Lifshitz\thanks{Einstein Institute of Mathematics, Hebrew University of Jerusalem.} \and
Dor Minzer\thanks{Institute for Advanced Study. Supported partially by NSF grant DMS-1638352 and Rothschild Fellowship.} \and
Muli Safra\thanks{School of Computer Science, Tel Aviv University. Supported by the European Research Council (ERC) under the European Union’s Horizon 2020 research and innovation programme (Grant agreement No. 835152).}}
\date{\vspace{-5ex}}
\usepackage{fullpage}
\usepackage{amsthm}
\usepackage{amsmath,amssymb,amsfonts,nicefrac}
\usepackage{xspace}
\usepackage{color}
\usepackage{url}
\usepackage{hyperref}
\usepackage{bm}
\usepackage{bbm}
\usepackage{times}

\usepackage{enumitem}
\usepackage{dsfont}

\newtheorem{thm}{Theorem}[section]

\newtheorem{lemma}[thm]{Lemma}
\newtheorem{corollary}[thm]{Corollary}
\newtheorem{claim}[thm]{Claim}

\newtheorem{definition}[thm]{Definition}
\newtheorem{remark}[thm]{Remark}

\newtheorem{fact}[thm]{Fact}

\newcommand\E{\mathop{\mathbb{E}}}
\newcommand\card[1]{\left| {#1} \right|}
\newcommand\sett[2]{\left\{ \left. #1 \;\right\vert #2 \right\}}

\newcommand\set[1]{{\left\{ #1 \right\}}}
\newcommand\Prob[2]{{\Pr_{#1}\left[ {#2} \right]}}

\newcommand\cProb[3]{{\Pr_{#1}\left[ \left. #3 \;\right\vert #2 \right]}}

\newcommand\Expect[2]{{\mathop{\mathbb{E}}_{#1}\left[ {#2} \right]}}

\newcommand\norm[1]{\| #1 \|}
\newcommand\power[1]{\set{0,1}^{#1}}

\newcommand\half{{\frac{1}{2}}}
\newcommand\quarter{{\frac{1}{4}}}

\newcommand\defeq{\stackrel{def}{=}}
\newcommand\skipi{{\vskip 10pt}}

\newcommand\inner[2]{\langle{#1},{#2}\rangle}
\newcommand\eps{\varepsilon}

\newcommand*\R{\mathbb{R}}

\renewcommand\geq{\geqslant}
\renewcommand\leq{\leqslant}

\newcommand{\cupdot}{\mathbin{\mathaccent\cdot\cup}}

\newcommand\indicator[1]{{\mathds 1}_{#1}}

\newcommand{\rom}[1]{\uppercase\expandafter{\romannumeral #1\relax}}

\makeatletter
\newtheorem*{rep@theorem}{\rep@title}
\newcommand{\newreptheorem}[2]{%
\newenvironment{rep#1}[1]{%
\def\rep@title{\bf #2 \ref*{##1} \text{(Restated)} }%
\begin{rep@theorem} }%
{\end{rep@theorem} } }

\newtheorem*{rep@claim}{\rep@title}
\newcommand{\newrepclaim}[2]{%
\newenvironment{rep#1}[1]{%
\def\rep@title{\bf #2 \ref*{##1} \text{(Restated)} }%
\begin{rep@claim} }%
{\end{rep@claim} } }

\newtheorem*{rep@lemma}{\rep@title}
\newcommand{\newreplemma}[2]{%
\newenvironment{rep#1}[1]{%
\def\rep@title{\bf #2 \ref*{##1} \text{(Restated)} }%
\begin{rep@lemma} }%
{\end{rep@lemma} } }

\makeatother
\newreptheorem{theorem}{Theorem}
\newrepclaim{claim}{Claim}
\newreplemma{lemma}{Lemma}

\begin{document}
\maketitle

\begin{abstract}
The total influence of a function is a central notion in analysis of Boolean functions,
and characterizing functions that have small total influence is one of the most fundamental
questions associated with it. The KKL theorem and the Friedgut junta theorem give a strong characterization
of such functions whenever the bound on the total influence is $o(\log n)$. However, both results become useless
when the total influence of the function is $\omega(\log n)$.
The only case in which this logarithmic barrier has been broken for an interesting class of
functions was proved by Bourgain and Kalai, who focused on functions that are symmetric under large
enough subgroups of $S_n$.

In this paper, we build and improve on the techniques of the Bourgain-Kalai paper and establish
new concentration results on the Fourier spectrum of Boolean functions with small total influence.
Our results include:
\begin{enumerate}
  \item A quantitative improvement of the Bourgain--Kalai result regarding the total influence of functions that are transitively symmetric.
  \item A slightly weaker version of the Fourier--Entropy Conjecture of Friedgut and Kalai. Our result establishes new bounds on the Fourier entropy of a Boolean function $f$, as well as stronger bounds on the Fourier entropy of low-degree parts of $f$. In particular, it implies that the Fourier spectrum of a constant variance, Boolean function $f$
  is concentrated on $2^{O(I[f]\log I[f])}$ characters, improving an earlier result of Friedgut. Removing the $\log I[f]$ factor would essentially resolve the Fourier--Entropy Conjecture, as well as
  settle a conjecture of Mansour regarding the Fourier spectrum of polynomial size DNF formulas.
\end{enumerate}

Our concentration result for the Fourier spectrum of functions with small total influence also has new implications
in learning theory. More specifically, we conclude that the class of functions whose total influence is at most $K$ is
agnostically learnable in time $2^{O(K\log K)}$ using membership queries. Thus,
the class of functions with total influence $O(\log n/\log \log n)$ is agnostically learnable in ${\sf poly}(n)$ time.
\end{abstract}

\section{Introduction}
The field of Analysis of Boolean functions is by now an integral part of Theoretical Computer Science, Combinatorics
and Probability. Many basic results, such as the KKL Theorem \cite{KKL} and the various junta theorems \cite{Friedgut98,Bourgain}
have a wide range of applications including PCP constructions \cite{Hastad,KKMO,DS}, metric non-embeddability results \cite{KhotNaor},
Extremal Combinatorics \cite{DinurFriedgut,KellerLifshitz} and many more.

Perhaps the most basic, non-trivial, question in the field is to characterize functions that have small total influence.
Throughout the paper, we will consider the Boolean hypercube $\power{n}$ equipped with the uniform measure, \footnote{Many
of the statements we give have natural analogs for the $p$-biased measure.} and in particular Boolean functions defined over it, i.e.\ $f\colon\power{n}\to\power{}$. The influence of the $i$th variable, $I_i[f]$, is the probability that $f(x) \neq f(x\oplus e_i)$
when we sample $x$ according to the uniform distribution.
The total influence of $f$ is the sum of all individual influences, i.e.\
$I[f] = I_1[f]+\ldots+I_n[f]$. What can be said about a function $f$ with constant variance, that has small total influence,
i.e.\ $I[f]\leq K$?

One obvious example of such functions are {\it K-juntas}, i.e.\ functions $f$ that depend on at most $K$ variables.
Clearly, if $f$ is a $K$-junta, then $I[f]\leq K$.
A better example, known as the Tribes function, was given by Ben-Or and Linial \cite{BenorLinial}, and it depends on $e^{\Omega(K)}$
variables, all of which have the same influences. The KKL Theorem and the Friedgut junta theorem state that in a sense, these are the worst possible examples.
The KKL Theorem \cite{KKL} asserts that in this case, there must be a variable $i$ with a large individual influence of $e^{-O(K)}$.
Friedgut \cite{Friedgut98} strengthened that result, showing that $f$ in fact must essentially depend only on $e^{O(K)}$ variables.
We note that these results are very effective when $K$ is thought of as constant, and remain meaningful as long as $K\leq \eps \log n$. When
$K$ is, say, $100\log n$, these results become completely trivial -- this is a well known barrier in analysis of Boolean functions, regarding which
very little is known.

\paragraph{Sharp thresholds.}
Motivated by the study of sharp thresholds of graph properties, Bourgain and Kalai \cite{BourgainKalai} studied the above question for functions
$f\colon\power{n}\to\power{}$ that are symmetric with respect to a subgroup $G\subseteq S_n$.\footnote{We say $f$ is symmetric under
a permutation $\pi\in S$, if $f(x) = f(\pi(x))$ for all $x\in\power{n}$, where $\pi(x)_i = x_{\pi(i)}$. We say $f$ is symmetric under a set of
permutations $G\subseteq S_n$ if it is symmetric under each $\pi\in G$.} They showed that if the subgroup $G$ is nice enough, then one
has $I[f] = \omega(\log n)$ for all functions $f$ symmetric under $G$. More precisely, for each $\eps>0$ and subgroup $G$, Bourgain and Kalai
consider the parameter $a_{\eps}(G)$, defined to be smallest possible $d$ such that the orbit of each $S\subseteq[n]$ of size $d$ is at least of size
$e^{d^{1+\eps}}$. Using this parameter, Bourgain and Kalai prove that $I[f]\geq C(\eps) a_{\eps}(G){\sf var}(f)$, for some $C(\eps)>0$.
The class of subgroups for which this statement is useful includes symmetries of graphs, hypergraphs, $GL(n,\mathbb{F}_q)$ and more. The latter result played an important part in a recent
result in coding theory, showing that Reed-Muller codes with constant rate achieve capacity over erasure channels with random errors \cite{RMcapacitiy} (though by
now, an alternative argument that bypasses the use of Bourgain-Kalai is known \cite{RMcapacity2}).

\paragraph{Learning theory.}
Proving strong structural results on Boolean functions with bounded total influence, say at most $100\log n$, if often times useful
in designing learning algorithms for specific concept classes, such as the class of polynomial size DNF formulas.
In particular, it has long been known that a concept class in which each function can be approximated by a sparse polynomial is
learnable with membership queries~\cite{GL,KM}. More recently, it has been shown that such property is strong enough to imply
learnability in the presence of some errors, i.e.\ in the agnostic learning model \cite{GKK}.
We elaborate on the topic and on our results along these lines in Section~\ref{sec:learning}.

\paragraph{Percolation.}
Another motivation to develop tools bypassing the logarithmic barrier comes from
percolation theory. Kalai~\cite{KalaiPersonal} asked whether there is
a variant of the Bourgain-Kalai Theorem in which the symmetry condition is relaxed to
a weaker notion of regularity.
His question was motivated by a problem in percolation
theory, in which one has a sequence of function $f_n$ (which is the indicator of the crossing
event in the $3$-dimensional grid at the critical probability) and the goal is to prove
good lower bounds on the total influence of $f_n$. More precisely, the goal is to
prove that for every $n$ one has $I[f_n]\geq a_n>0$, where the sequence $a_n$ satisfies
that $\sum\limits_{n=1}^{\infty}\frac{1}{n a_n}$ converges (i.e.\, morally that $a_n$ is slightly
larger than $\log n$). The class of functions $f_n$, however, does not have the symmetries
required for the Bourgain-Kalai Theorem.

Our main result can be viewed as a variant of the Bourgain-Kalai Theorem that relaxes the
symmetry condition, and we prove it is enough that all low Fourier coefficients of $f$ are small.

\subsection{The Fourier--Entropy Conjecture}
Another form of structural results on Boolean functions with $I[f]\leq K$ one may hope for, is
that of concentration of the Fourier Spectrum only on a small number of characters. In other words, can we say that except for a negligible mass, all Fourier weight
of $f$ is concentrated on few Fourier coefficients? Friedgut's theorem \cite{Friedgut98} (or rather, its proof) implies that except for negligible amount of mass, all
Fourier weight of $f$ lies on at most $e^{O(K^2)}$ Fourier coefficients; in general, this is the best bound known to date. Friedgut and Kalai
\cite{FriedgutKalai} conjectured that the actual answer should be $e^{O(K)}$; in fact, they propose the more refined {\it Fourier--Entropy Conjecture},
stating that the Shannon-Entropy of the Fourier spectrum of a Boolean function (thought of as a distribution) is at most $O(I[f])$.
Here, the Fourier entropy of a function $f$ is given by $H[\widehat{f}] = \sum\limits_{S\subseteq[n]}{\widehat{f}(S)^2\log(1/\widehat{f}(S)^2)}$.
Despite significant efforts, progress towards the Fourier--Entropy Conjecture has been slow, and it has been proved only for
special classes of functions \cite{CKSS,KLW,OT,OWZ,WWW}.

\skipi

The min-entropy of the Fourier Spectrum of a function
$f$ is defined by $H_{\infty}[\widehat{f}] = \min_{S}\log(1/\widehat{f}(S)^2)$.
The Fourier--Min--Entropy Conjecture is a relaxation of the Fourier--Entropy Conjecture, stating that
min-entropy of the Fourier spectrum of a Boolean function
is at most $O(I[f])$. As the min-entropy of a distribution is always upper bounded
by the Shannon-Entropy of a distribution, one sees that this conjecture is strictly weaker. O'Donnell noted that for monotone functions
(and more generally for unate functions), this conjecture follows immediately from the KKL Theorem (while the Fourier--Entropy Conjecture is not known for monotone functions).
\footnote{A function $f\colon \power{n} \to\power{}$ is said to be increasing (respectively decreasing) with respect to coordinate $i$, if for every $x\in\power{n}$ with $x_i = 0$,
it holds that $f(x\oplus e_i)\geq f(x)$ (respectively $f(x\oplus e_i)\leq f(x)$). The function $f$ is called monotone if it is increasing along each $i\in[n]$,
and called unate if along each $i\in[n]$ it is either increasing or decreasing.}
Progress towards this conjecture has also been slow \cite{ACKSW,Shalev}.

\subsection{Main results}
Our main results are new bounds on the Fourier min-entropy and the Fourier entropy of a Boolean function $f\colon\power{n}\to\power{}$.
It will be convenient to state them in terms of the normalized total influence of $f$, i.e.\ $\tilde{I}[f] = \frac{I[f]}{{\sf var}(f)}$.
First, we show that the Fourier min-entropy Conjecture holds up to poly-logarithmic
factor in $\tilde{I}[f]$.

\begin{thm}\label{thm:main}
  There is $C>0$, such that for any function $f\colon\power{n}\to\power{}$ there is a non-empty $S\subseteq[n]$
  of size $O(\tilde{I}[f])$ such that
  $\card{\widehat{f}(S)}\geq 2^{-C\cdot \card{S}\log(1+\tilde{I}[f])}\sqrt{{\sf var}(f)}$ . In particular,
  \[
  H_{\infty}[f]\leq O\left(\tilde{I}[f]\log(1+\tilde{I}[f])\right).
  \]
\end{thm}

Our result is in fact stronger in several ways. First, we show that not only there exists a significant Fourier coefficient for $f$,
but in fact almost all the Fourier mass of $f$ lies on such characters.

\begin{thm}\label{thm:main_gen}
  For every $\eta>0$, there exists $C>0$, such that for all $f\colon\power{n}\to\power{}$ we have
  \begin{equation}\label{eq:concentrate_improved}
  \sum\limits_{S} \widehat{f}(S)^2 1_{\card{\widehat{f}(S)}\leq 2^{-C\cdot \tilde{I}[f]\log(1+\tilde{I}[f])}}
  \leq \eta\cdot {\sf var}(f).
  \end{equation}
\end{thm}

Second, we show that a slightly weaker version of the Fourier--Entropy Conjecture holds for the low-degree parts of $f$.
Here, the part of $f$ of degree at most $d$ is denoted by $f^{\leq d}$ and is defined to be the part of
the Fourier transform of $f$ up to degree $d$ (see Section~\ref{sec:prelim} for a formal definition).
\begin{thm}\label{thm:main_entropy_improved}
  There exists an absolute constant $K>0$, such that for any $D\in\mathbb{N}$ and $f\colon\power{n}\to\power{}$ we have that
  \[
     H[\widehat{f^{\leq D}}]\leq K\sum\limits_{\card{S}\leq D}{\card{S}\log(\card{S}+1)\widehat{f}(S)^2} + K \cdot I[f].
  \]
\end{thm}
Note that as $I[f] = \sum\limits_{S}{\card{S}\widehat{f}(S)^2}$, Theorem~\ref{thm:main_entropy_improved} just falls short
of proving the Fourier entropy conjecture by a factor of $\log(\card{S})$. In the worst case, this factor may be as large as
$\log({\sf deg}(f))$, however for most interesting functions the contribution of this logarithmic factor is typically much smaller.
We also note that a particularly interesting setting of $D$ is $D = 100\cdot I[f]$, since most of the Fourier mass of $f$ lies
on characters $S$ such that $\card{S}\leq D$ and then the logarithmic factor contributes at most $\log(D)$.

\skipi

The above results follow from our main technical result,
Corollary~\ref{corr:main_set_param_improved2}, which may be of independent interest.
We remark that this project began with the goal of gaining a better understanding of the (notoriously hard) Bourgain-Kalai paper,
and explore the underlying idea that allowed them to bypass the logarithmic barrier.
Indeed, our proofs build and improve upon the ideas of
Bourgain and Kalai, and in our view are also considerably simpler.
The core idea of the argument now boils down to a general statement
that upper bounds inner products $\inner{f}{g}$ for a Boolean function $f$ and a real-valued, low-degree function $g$,
by their Fourier coefficients, total influences and norms. This result could be thought of a successive series of inequalities that improve each
other, the first one of which is the KKL Theorem, and the proof of each inequality uses the previous inequalities (formally, by induction).
We defer a more detailed discussion of our techniques and proofs to Section~\ref{sec:techniques}.

\begin{remark}
One can establish variants of Theorems~\ref{thm:main},~\ref{thm:main_gen} and~\ref{thm:main_entropy_improved}
for general real-valued functions $f\colon\power{n}\to\mathbb{R}$. The proof goes along the same lines, except that the quantity $I[f,g]$ considered
throughout the proof is redefined to be $I[f,g] = \sum\limits_{i=1}^{n}\norm{\partial_i f}_2^{1/4}\norm{\partial f}_{4/3}^{1/2}\norm{\partial_i g}_2$.
\end{remark}

\subsection{Applications}
\subsubsection{Learning Theory}\label{sec:learning}
The Fourier--Entropy Conjecture is closely to related
to the problem of learning functions in the membership model. Intuitively, if a function has small Fourier Entropy, that that means that its Fourier transform is concentrated on a few characters. In other words, it means that the function can be approximated by a sparse polynomial, a class that is very important in the context of learning theory.

It is well-known that the class of functions approximated by sparse polynomials is learnable using membership queries \cite{GL,KM}. However, one may wonder if this class is learnable in the more challenging model of agnostic learning.

The framework of agnostic learning was defined by Kearn et al.\cite{KSS}, who proposed it as a more realistic version of the PAC-learning model. Indeed, it is aimed at capturing the intuition that often in the task of learning a function, the data we see is in fact a noisy version of the actual data in the real world, and therefore one cannot make very strong structural assumptions on it, but rather that it is close to a very structured object. We formalize this notion next.

A noteable example in this context is the class of polynomials-size DNF formulas, which is known to be somewhat sparse (its Fourier spectrum is concentrated on at most on at most $n^{O(\log\log n)}$ characters) but not enough

Let $f\colon\power{n}\to\power{}$ be an \emph{arbitrary} Boolean
function and let $\mathcal{C}\subset \set{g\mid g\colon\power{n}\to\power{}}$ be a concept class. Define
${\sf opt}_{\mathcal{C}}(f) = \min_{c\in\mathcal{C}}\Prob{x\in\power{n}}{c(x)\neq f(x)}$, i.e.\
the distance of $f$ from the concept class $\mathcal{C}$.

\begin{definition}
We say that $\mathcal{C}$ is \emph{agnostically learnable} with queries with respect to the uniform distribution, if there exists a randomized algorithm that given black box access to any $f$, runs in time $poly(n,\eps^{-1},\log\frac{1}{\delta})$
and outputs a hypothesis $h$ such that
$\Prob{x\in\power{n}}{h(x) \neq f(x)} \leq  {\sf opt}_{\mathcal{C}}(f) +\eps$  with probability $1-\delta$.
\end{definition}

Golpalan et al.\cite{GKK} showed that every concept class that has sparse approximation is agnostically learnable (the running time depends on the sparsity of the approximation and other parameters).
\begin{thm}[\cite{GKK}, Theorem 16]
Let $t,\eps>0$ and suppose $\mathcal{C}$ is a concept class such that for every $c\in \mathcal{C}$, there exists a $p\colon\power{n}\to\mathbb{R}$ such that
\begin{enumerate}
    \item The polynomial $p$ is $t$-sparse, i.e.\ $\sum\limits_{S}\card{\widehat{p}(S)}\leq t$.
    \item The polynomial $p$ approximates $c$ in $\ell_1$: $\Expect{x\in\power{n}}{\left|p(x)-c(x)\right|} \leq  \eps$.
\end{enumerate}
Then there exists a agnostic learning algorithm for $\mathcal{C}$ that runs in time ${\sf poly}(t,\frac{1}{\eps},n)$.
\end{thm}

Combining this result with Theorem \ref{thm:main_gen}, we get the following corollary.
\begin{corollary}\label{corr:learning_gen}
For every $\eps>0$ and $K\in\mathbb{N}$,
the concept class $\mathcal{C}$ of functions whose total influence is at most $K$ is agnostically learnable in time
${\sf poly}\left(2^{O_{\eps}(K\log K)},\frac{1}{\eps},n\right)$
\end{corollary}
Indeed, given $\eps>0$ we apply Theorem~\ref{thm:main_gen} with $\eta=\eps^2$, and get the Fourier mass of $f$ outside
\[
\mathcal{S} = \sett{S}{\card{\widehat{f}(S)}\geq 2^{-C K\log K}}
\]
is at most $\eps^2$, for some $C$ depending on $\eps$. Therefore the polynomial
$p(x) = \sum\limits_{S\in\mathcal{S}}\widehat{f}(S)\chi_S(x)$ is close to $f$ in $\ell_2$, i.e.\ $\norm{f-p}_2\leq \eps$, and in particular $\norm{f-p}_1\leq \eps$. Secondly, since the sum of
squares of Fourier coefficients of $f$ is at most $1$, we have that $\card{\mathcal{S}}\leq 2^{2CK \log K}$, and in particular $p$ is $t$-sparse for $t=2^{2CK\log K}$.

Corollary~\ref{corr:learning_gen} implies that as far as polynomial time algorithms are concerned, one can agnostically learn the concept class of functions whose total influence is
$O\left(\frac{\log n}{\log\log n}\right)$.
This just falls short of capturing the class of polynomial size DNF formulas, which is known to be contained in the class of functions with total influence $O(\log n)$. We remark that while learning algorithms for the class of polynomial size DNF's are known \cite{Jackson}, no agnostic learner is known.
Proving the Fourier Entropy Conjecture is a known avenue towards achieving this goal~\cite{GKK2}, and
we believe this avenue to be promising in light of our results. More concretely, to achieve this goal it would be enough to ``shave off'' the logarithmic factor from Theorem~\ref{thm:main_gen} (which would also establish a conjecture of Mansour~\cite{MansourConj}).

\subsubsection{Sharp thresholds}
Theorem~\ref{thm:main} can be used to prove a quantitatively improved, nearly tight Bourgain-Kalai-like Theorem.
For example, it implies that graph properties with constant variance have total influence at least
$\Omega\left(\frac{\log^2\ n}{(\log\log n)^2}\right)$. It is also strong enough to imply polynomial lower bound
on the total influence of a function $f$, provided its group symmetries is large enough.
See Section~\ref{sec:BK_deduce} for more details.

\paragraph{Organization.}
In Section~\ref{sec:prelim} we give standard tools and notions from discrete Fourier analysis.
In Section~\ref{sec:rest,part,deg} we present two important ideas that are used in the proof of our main results.
In Section~\ref{sec:main_pfs} we state and prove a basic version of our main technical result, Theorem~\ref{thm:main_formal_fourier}
and Corollary~\ref{corr:main_set_param}.
In Section~\ref{sec:improved_results} we give a first improvement over Theorem~\ref{thm:main_formal_fourier}, namely Theorem~\ref{thm:main_formal_fourier_improved}
and Corollary~\ref{corr:main_set_param_improved}, and in Section~\ref{sec:improved_results2} we give a further improvement of that result,
namely Theorem~\ref{thm:main_formal_fourier_improved2} and Corollary~\ref{corr:main_set_param_improved2}
Finally, in Section~\ref{sec:imp_corollaries} we deduce several corollaries, including Theorems~\ref{thm:main},~\ref{thm:main_gen} and~\ref{thm:main_entropy_improved}.

\section{Preliminaries}\label{sec:prelim}
In this section we describe the basics of Fourier analysis over the hypercube that will be needed in this paper
(see \cite{Odonnell} for a more systematic treatment).
\paragraph{Notations.} We denote $[n] = \set{1,\ldots,n}$. Throughout the paper, $\log x$ is the natural logarithm of $x$.

\subsection{Fourier analysis on the hypercube}
Consider the hypercube $\power{n}$ along with the uniform measure $\mu$, and consider real-valued functions
$f\colon\power{n}\to\mathbb{R}$ equipped with the inner product $\inner{f}{g} = \Expect{x\sim \mu}{f(x)g(x)}$.
The set $\set{\chi_S}_{S\subseteq[n]}$, where $\chi_S = (-1)^{\sum\limits_{i\in S} x_i}$ is the well-known
Fourier basis that forms an orthonormal basis with respect to our inner product, thus one can expand any
$f\colon\power{n}\to\mathbb{R}$ as
\[
f(x) = \sum\limits_{S\subseteq[n]}{\widehat{f}(S) \chi_S(x)}, \qquad\qquad \text{where }\widehat{f}(S) = \inner{f}{\chi_S}.
\]
The degree of a function $f$ is ${\sf deg}(f) = \max_{S:\widehat{f}(S)\neq 0}{\card{S}}$.
Since $\set{\chi_S}$ is an orthogonal system, we have the Parseval/ Plancherel equality.

\begin{fact}
  For any $f,g\colon\power{n}\to\mathbb{R}$ we have that $\inner{f}{g} = \sum\limits_{S\subseteq[n]}{\widehat{f}(S)\widehat{g}(S)}$.
\end{fact}

\subsection{Restrictions, derivatives and influences}
\vspace{-1ex}
Given a function $f\colon\power{n}\to\mathbb{R}$, a set of variables $S\subseteq[n]$ and $z\in \power{S}$,
the restricted function $f_{S\rightarrow z}$ is the function from $\power{[n]\setminus S}$ to $\mathbb{R}$
resulting from fixing $S$'s coordinates in $x$ to be $z$. If $S$ is a singleton $\set{i}$, we will denote
this restriction by $f_{i\rightarrow z}$.
\begin{definition}
  The discrete derivative of $f\colon \power{n}\to\mathbb{R}$ in direction $i$ is a function
  $\partial_i f\colon\power{n\setminus[i]}\to\mathbb{R}$ defined by $\partial_i f(x) = \half(f_{i\rightarrow 0}(x) - f_{i\rightarrow 1}(x))$.

  More generally, for a set of variables $T\subseteq[n]$, the derivative of $f$ with respect to $T$ is $\partial_T f\colon\power{[n]\setminus T}\to\mathbb{R}$
  is defined by iteratively applying the derivative operator on $f$ for each $i\in T$. Alternatively,
  \[
  \partial_T f(x) = 2^{-\card{T}}\sum\limits_{z\in \power{T}}(-1)^{\card{z}}f_{T\rightarrow z}(x).
  \]
\end{definition}
\noindent The Fourier expansion of $\partial_T f(x)$ is
$\sum\limits_{S\supseteq T}{\widehat{f}(S)\chi_{S\setminus T}(x)}$.

The following definition generalizes the notion of influences to real-valued functions. We remark that for Boolean functions,
it differs by a factor of $4$ from the definition given in the introduction (this is done only for convenience).
\begin{definition}
  Let $f\colon\power{n}\to\mathbb{R}$ be a function.
  The influence of a variable $i\in [n]$ is given by
  $I_i[f] = \norm{\partial_i f}_2^2$, and the total influence
  of $f$ is defined to be $I[f] = \sum\limits_{i=1}^{n} I_i[f]$.

  The generalized influence of a set $S\subseteq[n]$ on $f\colon\power{n}\to\mathbb{R}$ is $I_S[f]=\norm{\partial_S f}_2^2$.
\end{definition}
\noindent
Using the Fourier expression for $\partial_T f$ and Parseval,
we see that $I_T[f] = \sum\limits_{S\supseteq T}{\widehat{f}(S)^2}$.
In particular, using this formula for $T$'s that are singletons and summing, one gets that the total influences of $f$ can be written as
$I[f] = \sum\limits_{S}{\card{S}\widehat{f}(S)^2}$.

\paragraph{Low-degree part and low-degree influences.} For $f\colon\power{n}\to\mathbb{R}$
and $d\leq n$, we define the degree at most $d$ part of $f$, $f^{\leq d}$, to be
\[
f^{\leq d}(x) = \sum\limits_{\card{S}\leq d}{\widehat{f}(S)\chi_S(x)}.
\]
Using the Fourier formula for the total influence and Parseval, one sees that
$I[f^{\leq d}] \leq d\norm{f}_2^2$, and hence there are at most $\frac{d}{\tau}$ variables $i$ that have $I_i[f^{\leq d}]\geq \tau \norm{f}_2^2$. A similar
property holds for generalized low-degree influences of $f$.

\begin{fact}\label{fact:sum_gen_inf}
  For any $v\leq d$ and $f\colon\power{n}\to\mathbb{R}$, one has
  that $\sum\limits_{\card{T} \leq v} I_T[f^{\leq d}]\leq 2 d^{v} \norm{f}_2^2$.
\end{fact}
\begin{proof}
  Using the Fourier formula for $I_S[f^{\leq d}]$, the left hand side is equal to
  \[
  \sum\limits_{\card{T} \leq v} I_T[f^{\leq d}]
  =\sum\limits_{\card{T} \leq v} \sum\limits_{\substack{\card{S}\leq d \\ S\supseteq T}} \widehat{f}(S)^2
  =\sum\limits_{\card{S}\leq d}\widehat{f}(S)^2 \sum\limits_{\substack{T\subseteq S\\ \card{T}\leq v}} 1.
  \]
  Fix $S$. The inner summation is equal to the number of subsets of $S$ of size at most $v$, hence it is
  equal to ${\card{S} \choose 0} + \ldots + {\card{S} \choose \min(v,\card{S})}
  \leq \sum\limits_{k=0}^{v}\frac{\card{S}^k}{k!}$. If $v\leq 1$, then this sum is at most
  $2 d^v$, and if $v > 1$, we may upper bound the sum by
  $1 + d + \sum\limits_{k=2}^{v}\frac{\card{S}^k}{k!}
    \leq 1+d+d^v\sum\limits_{k=2}^{\infty}\frac{1}{k!}
    = 1 + d + d^v(e-2)
    \leq 2 d^v$.
\end{proof}

\subsection{Random restrictions}
Let $f\colon\power{n}\to\mathbb{R}$ be a function, and $I\subseteq[n]$. A random restriction of $f$
on $I$ is the function $f_{I\rightarrow z}$ where we sample $z$ uniformly from $\power{I}$.
In our applications we will usually have two functions, $f,g\colon\power{n}\to\mathbb{R}$ and we will
consider the effect of the same random restriction of them. For example, it is easy to show that for
any $I\subseteq[n]$, the expected inner product of $\inner{f_{I\rightarrow z}}{g_{I\rightarrow z}}$
over $z\in_R\power{I}$ is equal to $\inner{f}{g}$. Another quantity associated with $f,g$
that we will consider is the cross-total-influence.
\begin{definition}
  For any $f,g\colon\power{n}\to\mathbb{R}$ and $i\in [n]$, we define the cross-influence
  along direction $i$ to be $I_i[f,g] = \sqrt{I_i[f] I_i[g]}$.
  The cross-total-influence of $f,g$ is given by $I[f,g] = \sum\limits_{i=1}^{n} I_i[f,g]$.
\end{definition}
By Cauchy-Schwarz, one always has that $I[f,g]\leq \sqrt{I[f]I[g]}$. The quantity $I[f,g]$
though will be easier for us to work with inductively, and the following property will be useful
for us.

\begin{lemma}\label{lem:cross_inf}
  For any $f,g\colon\power{n}\to\mathbb{R}$ and $I\subseteq [n]$, we have that
  \[
  \Expect{z\in\power{I}}{I[f_{I\rightarrow z},g_{I\rightarrow z}]}\leq
  \sum\limits_{i\not\in I}{I_i[f,g]}.
  \]
\end{lemma}
\begin{proof}
By definition, the left hand side is equal to
  \begin{align*}
  \Expect{z\in\power{I}}{\sum\limits_{i\in\overline{I}} \sqrt{I_i[f_{I\rightarrow z}]} \cdot \sqrt{I_i[g_{I\rightarrow z}]}}
  &=\sum\limits_{i\in\overline{I}} \Expect{z\in\power{I}}{\sqrt{I_i[f_{I\rightarrow z}]} \cdot \sqrt{I_i[g_{I\rightarrow z}]}}\\
  &\leq \sum\limits_{i\in\overline{I}} \sqrt{\Expect{z\in\power{I}}{I_i[f_{I\rightarrow z}]}} \cdot
  \sqrt{\Expect{z\in\power{I}}{I_i[g_{I\rightarrow z}]}}\\
  &=\sum\limits_{i\in\overline{I}} \sqrt{I_i[f]}\sqrt{I_i[g]},
  \end{align*}
  where the second transition is by Cauchy-Schwarz.
\end{proof}

We will also need the following fact about Fourier coefficients of random restrictions.
\begin{lemma}\label{lem:rr_coef}
  For any $g\colon\power{n}\to\mathbb{R}$, $I\subseteq [n]$ and $S\subseteq I$ we have that
  \[
  \Expect{z\in\power{I}}{\widehat{g_{J\rightarrow z}}(S)^2}
  = \sum\limits_{T: T\cap I = S}{\widehat{g}(T)^2}.
  \]
\end{lemma}
\begin{proof}
The Fourier coefficient of $S$ in $g_{J\rightarrow z}$ is $\sum\limits_{T\subseteq \overline{I}}\widehat{g}(S\cup T)\chi_T(z)$.
Thinking of the latter as a function of $z$ and using Parseval's equality, we get that the expectation of its square is equal to
$\sum\limits_{T\subseteq\overline{I}}{\widehat{g}(S\cup T)^2}$.
\end{proof}

\subsection{Hypercontractivity}
We will need the hypercontractive inequality \cite{Beckner,Bonami,Gross}, which states that for $q\geq 2$, the $q$-norm and $2$-norm
of degree-$d$ functions is comparable up to exponential factor in $d$.
\begin{thm}\label{thm:hypercontractivity}
If $f\colon\power{n}\to\R$ is a function of degree at most $d$, and $q\geq 2$,
then $\norm{f}_q\leq (q-1)^{d/2}\norm{f}_2$.
\end{thm}

\section{Restrictions, partitions and degree reductions}\label{sec:rest,part,deg}
In this section we state several lemmas that will be helpful in the proof of Theorem~\ref{thm:main},
and we start by describing the basic motivation.

Using hypercontractivity (Theorem~\ref{thm:hypercontractivity}) incurs a loss of exponential factor in
the degree of the function it is applied on. This exponential-factor loss is in fact the bottleneck in KKL and Friedgut's
Theorems. It is ultimately the reason it is hard to bypass the logarithmic barrier. Therefore one may search for
methods by which the degree of a function could be reduced prior to applying hypercontractivity.

One natural idea is to consider random restrictions: if we pick $I\subseteq[n]$ randomly by including each element with
probability $\half$ in it and restrict variables outside $I$, then the degree of $f$ shrinks by a factor of $2$ -- at least if we are willing to discard characters of
small total mass from it. Another point that is often useful, is that this allows one to view the given function $f$ as $f(y,z)$, where $y\in\power{I}$,
$z\in\power{\bar{I}}$ and the degree of $f$ on each one of $y,z$ (separately) is at most $d/2$.

Sometimes, it is necessary to get a degree reduction by more than a constant factor. One can certainly decrease the size of the set of live variables
$I$, however this introduces asymmetry between $I$ and $\bar{I}$, and thus we may not enjoy any reduction on $\bar{I}$. The idea of {\it random partition}
remedies this situation. To get a degree reduction by factor $m$, we may consider a random partition of $[n]$ into
$m$ disjoint sets, i.e.\, $[n] = I_1\cupdot\ldots\cupdot I_m$ generated by including every $i\in [n]$ in each one of them with equal probability.
We show that under mild conditions on $d$ and $m$, given a function $f$ there is a partition of $[n]$ into $m$ parts and a function $f'$ close to $f$ in $\ell_2$,
such that the restrictions of $f'$ to each one of the parts, $(f')_{\bar{I_j} \rightarrow z}$, is of degree (roughly) at most $d/m$.

\subsection{Random partitions}\label{sec:random_parts}

Recall that a function $f\colon\power{n}\to\mathbb{R}$ is called degree-$d$ homogenous
if all Fourier characters in its support have size exactly $d$, and we would like to
generalize it to functions that are ``almost'' degree $d$-homogenous.
We say that $S\subseteq [n]$ is \emph{of size $d$ within factor $\alpha$} if
$\alpha d\leq\card{S}\leq d$, and often write it succinctly by $\card{S}\sim d$
($\alpha$ will be clear from the context).
\begin{definition}
A function $f\colon\power{n}\to\R$ is called $(\alpha,d)$-almost-homogenous
if all Fourier characters $\chi_S$ in its support satisfy $\alpha d\leq \card{S}\leq d$.
\end{definition}

\begin{definition}
A partition of $[n]$ into $m$ parts is $\mathcal{I}=(I_1,\dots, I_m)$, where $I_1,\dots,I_m$ are pairwise disjoint sets that cover $[n]$.
\end{definition}
As discussed earlier, a random partition into $m$ parts is constructed by starting out with
$I_1=\ldots=I_m=\emptyset$, and then for each $i\in[n]$ choosing the part $I_j$ to which
we add $i$ uniformly among the $m$-parts.

The following claim asserts that if $S$ is of size roughly $d$, and we choose a random partition $\mathcal{I}$,
then with high probability $\card{S\cap I_j}$ is roughly of size $d/m$ for all $j\in[m]$ (for technical reasons,
since the definition of ``almost-homogenous'' allows for $\alpha$ to enter only in the lower bound on $S$,
we allow for a slack of $(1+\eps)$ factor in the size too).
\begin{lemma}\label{lem:random_partition}
Let $m,d\in\mathbb{N}$ and $\eps,\alpha\in(0,1)$.
If $S\subseteq [n]$ is of size $d$ within factor $\alpha$, then
\[
\Prob{\mathcal{I}=(I_1,\ldots,I_m)}{S\cap I_j\text{ is of size $(1+\eps)d/m$ within factor $(1-2\eps)\alpha$}}
\geq 1-2e^{-\frac{\eps^2\card{S}}{3m}}.
\]
\end{lemma}
\begin{proof}
Let $\mathcal{I}=(I_1,\ldots,I_m)$ be a random partition of $[n]$ into $m$ parts,
and for each $i\in [n]$ and $j\in[m]$ denote by $\indicator{i\in I_j}$ the indicator function of the event that $i$ is in $I_j$.
Thus, for each $j\in [m]$ we may write the random variable $\card{S\cap I_j}$ as a sum of independent random
variables $\sum\limits_{i\in S}{\indicator{i\in I_j}}$. Note that by linearity of expectation, we have that its expectation
is $\card{S}/m$, and we use Chernoff's bound to argue it is close to its expectation with high probability.

More precisely, using Chernoff's inequality for indicator random variables we have that
\[
\Prob{\mathcal{I}}{\left|\card{S\cap I_j} - \frac{\card{S}}{m}\right|\geq \eps\frac{\card{S}}{m}}
\leq 2e^{-\frac{\eps^2\card{S}}{3m}},
\]
and therefore by the Union Bound the probability $\left|\card{S\cap I_j} - \frac{\card{S}}{m}\right|\geq \eps\frac{\card{S}}{m}$
for some $j\in [m]$ is at most $m$ times that.
\end{proof}

Next, we use the above lemma to prove that if $f$ is almost $d$-homogenous, and we take a random partition $\mathcal{I}$,
then on each one of the parts $I_j$, $f$ is almost $d/m$-homogenous, provided we are willing to discard characters of
small total mass from the Fourier transform of $f$.

More formally, let $\alpha,\eps\in (0,1)$ and $d\in\mathbb{N}$ be parameters, and let $\mathcal{I} = (I_1,\ldots,I_m)$ be a partition of $[n]$.
We denote by  $G(\mathcal{I})$ the set of all $S\subseteq [n]$ such that $\card{S}\sim d$, and $S\cap I_j$
is of size $(1+\eps)d/m$ within factor $(1-2\eps)\alpha$ for all $j\in [m]$. In this language, the previous lemma states that provided that $d$
is large enough in comparison to $m$, for each $S$ of size $d$ within factor $\alpha$ we have that $\Prob{\mathcal{I}}{S\in G(\mathcal{I})}$
is close to $1$.

\begin{corollary}\label{cor:most_weight_on_good_charachters}
Let $m,d\in\mathbb{N}$, $\alpha,\eps\in(0,1)$, and let $f,g\colon\power{n}\to\R$ be functions such that
 $g$ is $(\alpha,d)$-almost-homogenous. Then there is a partition $\mathcal{I} = (I_1,\ldots,I_m)$ such that
\[
\sum\limits_{S\in G(\mathcal{I})}\card{\widehat{f}(S)\widehat{g}(S)}
\geq
\left(1-2me^{-\frac{\eps^2\alpha d}{3m}}\right)
\sum\limits_{S\subseteq [n]}\card{\widehat{f}(S)\widehat{g}(S)}.
\]
\end{corollary}
\begin{proof}
  We choose a partition $\mathcal{I}$ randomly, and lower bound the expectation of the right hand side.
  Let $\indicator{S\in G(\mathcal{I})}$ be the indicator random variable of $S$ being in $G(\mathcal{I})$.
  By Lemma \ref{lem:random_partition}, we have
  $\Expect{}{\indicator{S\in {\cal G}(\mathcal{I})}}\geq 1-2me^{-\frac{\eps^2\alpha d}{3 m}}$ for each $S$
  in the support of $\widehat{g}$, and therefore
  \[
  \Expect{\mathcal{I}}{\sum_{S\in G(\mathcal{I})}\card{\widehat{f}(S)\widehat{g}(S)}}
  =\sum\limits_{S\subseteq [n]}\card{\widehat{f}(S)\widehat{g}(S)}\Expect{\mathcal{I}}{\indicator{S\in G(\mathcal{I})}}
  \geq \left(1-2me^{-\frac{\eps^2\alpha d}{2m}}\right)\sum_{S\subseteq [n]}\card{\widehat{f}(S)\widehat{g}(S)}.
  \]
  In particular, there exists a choice of $\mathcal{I}$ for which the expression under the expectation on the
  left hand side is at least the right hand side.
\end{proof}

\subsection{Exchanging maximums and expectations}
We next discuss a tool that goes in handy with partitions. Let $f\colon\power{n}\to\mathbb{R}$, let $I\subseteq[n]$
and consider the Fourier coefficients of the restricted function, i.e.\ for each $S\subseteq I$ consider
$h_S\colon\power{\bar{I}}\to\mathbb{R}$ defined by $h_S(x) = \widehat{f_{\bar{I}\rightarrow x}}(S)$. Recall that
we are using restrictions (or more generally partitions) as a way to decrease the degree of the function, but
eventually we want to transfer the information we got on the restrictions back to information about $f$.
For example, if the bound we proved involves $2$-norms of the Fourier coefficients of the restrictions, i.e.\ of $h_S$'s,
then a corresponding bound using Fourier coefficients of $f$ can be established by Lemma~\ref{lem:rr_coef}.
The bound however could depend on the functions $h_S$ in a more
involved way, e.g.\ on other $\ell_p$-norms of them, in which case one can often get effective bounds using hypercontractivity.

Since we are interested in the min-entropy of a function $f$ (e.g.\ for Theorem~\ref{thm:main}), we will naturally
wish to understand the maximum Fourier coefficient after restriction as a function of $x$, $\max_{S} h_S(x)$, and relate
it to some parameters of $f$. We do that in Lemma~\ref{lem:expectation+maximus_exchange}.

Generalizing the above discussion, let $h_1,\dots , h_k\colon\power{n}\to\mathbb{R}$ be functions of low-degree,
and consider the following two quantities.
The first one is $\E_{x\sim \mu}\max_{i} h_i(x)^2$,
and in the second quantity we interchange
the order of these two operations, i.e.\ $\max_i \E_{x\sim \mu} h_i(x)^2$.
Note that for every $x$ and $j$ we have that
$\max_{i} h_i(x)^2 \geq h_j(x)^2$. Taking an expectation of this inequality over $x\sim \mu$,
and then maximum over $j$ establishes that the first quantity is always larger
than the second quantity. The following lemma asserts that these two quantities
are polynomially related, provided that the expected value of $\sum\limits_{i=1}^{k} h_i(x)^2$ is constant.
\begin{lemma}\label{lem:expectation+maximus_exchange}
 If $h_1,h_2,\dots,h_k\colon \power{n}\to\R$ are all of degree at most $d$, then
  \[
  \Expect{x}{\max_{i\in[k]} h_i(x)^2}
  \leq 3^d\max_{i\in[k]}\norm{h_i}_2
  \left(\Expect{x}{\sum\limits_{i=1}^{k} h_i(x)^2}\right)^{1/2}.
  \]
  \end{lemma}
  \begin{proof}
  Note that for all $x$, we have that $\max_{i\in[k]} h_i(x)^2\leq \left(\sum\limits_{i=1}^{k}h_i(x)^4\right)^{1/2}$,
  and it will be more convenient for us to upper bound the expectation of the latter function. By Jensen's inequality, its
  expectation is at most
   \begin{equation}\label{eq1}
   \left(\Expect{x}{\sum\limits_{i=1}^{k}h_i(x)^4}\right)^{\half}
   = \left(\sum\limits_{i=1}^{k}\Expect{x}{h_i(x)^4}\right)^{\half}.
   \end{equation}
 Using Theorem~\ref{thm:hypercontractivity}, we may upper bound
 $\Expect{x}{h_i(x)^4}$ by $9^d \cdot\Expect{x}{h_i(x)^2}^2$, and thus
 \[
 \eqref{eq1}
 \leq 3^d \left(\sum_{i\in[k]}\Expect{x}{h_i(x)^2}^2\right)^{\half}
 \leq 3^d \left(\max_{i\in[k]}{\Expect{x}{h_i(x)^2}}\sum_{i\in[k]}\Expect{x}{h_i(x)^2}\right)^{\half},
 \]
 and using the definition of $2$-norm completes the proof.
  \end{proof}

\section{A basic version of our main technical result}\label{sec:main_pfs}
In this section, we state and prove Theorem~\ref{thm:main_formal_fourier} and Corollary~\ref{corr:main_set_param},
which are basic, less quantitatively efficient forms of our main technical result. We find it natural though to
present them along with the slightly more natural argument, and encourage the reader to read this section before
moving on to Section~\ref{sec:improved_results}.

\subsection{Proof idea}\label{sec:techniques}
Before we prove (or even state) our main technical result, we begin with an informal overview of the idea.
We start with presenting the (one-line) proof of the KKL Theorem.
Let $f\colon\power{n}\to\power{}$, $d\in\mathbb{N}$, and
$g = f^{\leq d} - \widehat{f}(\emptyset)$; we have that:
\begin{align*}
\inner{f}{g}\leq
\sum\limits_{i=1}^{n} \inner{\partial_i f}{\partial_i (f^{\leq d})}
\leq \sum\limits_{i=1}^{n} \norm{\partial_i f}_{4/3}\norm{\partial_i (f^{\leq d})}_{4}
&\leq \sqrt{3}^d\sum\limits_{i=1}^{n} \norm{\partial_i f}_{4/3}\norm{\partial_i (f^{\leq d})}_{2}\\
&\leq 2\sqrt{3}^d I[f] \max_i I_i[f^{\leq d}]^{1/4},
\end{align*}
where in the second inequality we used H\"{o}lder's inequality and in the third inequality we used Theorem~\ref{thm:hypercontractivity}.
In the last inequality, we upper bounded $\norm{\partial_i (f^{\leq d})}_{2}$ by
$\norm{\partial_i f}_2^{1/2}\cdot \norm{\partial_i (f^{\leq d})}_{2}^{1/2} = I_i[f]^{1/4} \norm{\partial_i (f^{\leq d})}_{2}^{1/2}$,
and used the Booleanity of $f$ to bound $\norm{\partial_i f}_{4/3}\leq 2 I_i[f]^{3/4}$. Hence, this inequality gives us very good bounds on the
weight $f$ has on its low-degrees, provided that all of its low degree influences are small. This raises two questions:
\begin{enumerate}
  \item What bounds could be proved if we know that there are only few influential variables?
  \item Can we improve on this bound if we know stronger information about $f$, e.g.\ that its generalized low-degree influences,
  $I_S[f^{\leq d}]$, are all very small?
\end{enumerate}

For the first question, a standard bound proceeds by handling characters $S$ that consist only of influential variables separately
(similar to Lemma~\ref{lem:base_case} below). In essence, the bound says that if there are at most $T$ variables with influence at least
$\delta$, then there are at most $T^{d}$ characters consisting only of these influential variables, and one gets the bound
$T^d\max_{S}\card{\widehat{f}(S)\widehat{g}(S)} + \sqrt{3}^d I[f] \delta^{1/4}$, which in this case is just
\begin{equation}\label{eq:overview}
T^d\max_{S}\widehat{f}(S)^2 + \sqrt{3}^d I[f] \delta^{1/4}.
\end{equation}

The second question is more interesting, and the obvious naive attempt one may first have quickly fails.\footnote{Namely, the attempt is to run the proof of the KKL theorem with the generalized
derivatives instead of standard derivatives. This fails since the sum of the generalized derivatives in general may be much higher
than the total influence: if we consider say, order $v$ derivatives, then a character $S$ would be counted by ${\card{S}\choose v}$
generalized influences.}
One key insight in \cite{BourgainKalai}, is that a better way to use the information on generalized influence
is by relating them to Fourier coefficients of random restrictions (and importantly, to something slightly stronger, see Lemma~\ref{lem:rr_coef}).
For a function $f\colon\power{n}\to\power{}$, a set $I\subseteq[n]$ and a subset $S\subseteq I$,
the expected Fourier coefficient squared, $\widehat{f}_{I\rightarrow z}(S)^2$, is at most $I_S[f]$; also, the degree of a random restriction
(roughly speaking) of $f^{\leq d}$ is significantly smaller than $d$. Thus, by first applying
a random restriction, and then using inequality~\ref{eq:overview} (the ``KKL bound''), one may expect to get a meaningful bound for
the second question above --- this is indeed the case.
An important technical point is that one may ``switch'' the order of maximum and expectation in this argument, which is where Lemma~\ref{lem:expectation+maximus_exchange}
comes in handy.

The above argument gives a stronger bound than inequality~\eqref{eq:overview}, provided all of the generalized influence of $f$ (of some order) are very small;
this is very much analogous to the KKL Theorem we started with!
Again, the following question arises: can we base on this result a similar bound to inequality~\eqref{eq:overview}
in case we know $f$ only has a few noticeable generalized influences?
This is next step in the argument, and is established in a similar manner to the way we established inequality~\eqref{eq:overview}.

The final statement, Theorem~\ref{thm:main_formal_fourier} below, is the outcome of applying this idea inductively. There are several technical points omitted
from the above description that need to be taken into account to make it precise. To get a strong enough statement, the set of
live variables $I$ should be chosen randomly, but at the same time the degree of $f^{\leq d}$ under random restrictions has to decrease.
While this can probably be done, it is likely to be messy, and we bypass it by using random-partitions from Section~\ref{sec:random_parts}.
We then discard from the Fourier expansion of $f^{\leq d}$ characters $S$ that remain large with respect to one of the parts
(and argue they do not contribute too much to $\inner{f}{f^{\leq d}}$). Note that after discarding, it will no longer be true that we are working
with a function and its low degree part, so to facilitate induction we work with the two-function version of the problem instead. We then look
at each part $I_j$ of the partition $\mathcal{I}$, consider random restrictions of the discarded version of the low-degree part of $f$ and
apply the induction hypothesis on them (note that these restrictions also decrease the degree of the function considerably, as the partition $\mathcal{I}$ is picked
according to Corollary~\ref{cor:most_weight_on_good_charachters}). A slightly more detailed overview of the inductive step is given in Section~\ref{sec:inductive_step}.

\subsection{Statement of the main technical result}
In this section, we prove our main technical result. In the following statement, we have $\eps>0$,
and an increasing sequence of integers $d_0,\ldots,d_{k}$, and we will be interested in $S$ that are of size
$d_k$ within factor $\alpha$; it will be convenient for us to write it more succinctly as $\card{S}\sim d_k$.

In the statement below, one should think of $g$ (roughly) as $f^{\leq d}$, and the goal is to prove
that if $f$ has no large Fourier coefficients, then $\inner{f}{g} = o(\norm{f}_2^2)$, hence $f$
is concentrated on high degrees and in particular $I[f]\geq (1-o(1))d\norm{f}_2^2$.
The more general statement with general $g$ is crucial
for the inductive proof to go through (as hinted in the overview above).

%

\begin{thm}\label{thm:main_formal_fourier}
Let $k\in\mathbb{N}$, $\alpha,\eps\in(0,1)$,
let $d_0, d_1,\ldots,d_{k}$ be an increasing sequence such that
$d_0 = 1$ and for each $3\leq j\leq k$ we have $d_{j-1}\geq \frac{3}{\alpha\eps^2(1-2\eps)^{k}}\log(4(1+\eps)d_{j}/d_{j-1})$,
and let $0<\delta_k\leq\ldots\leq \delta_1$. Then there are $C_1,C_2, C_3$ specified below, such that the following holds.

If $f\colon \power{n}\to\power{}$ is a Boolean function, and $g\colon\power{n}\to\R$ is $(\alpha,d_k)$-almost-homogenous,
then
\begin{equation}\label{eq:thm_statement}
  \inner{f}{g} \leq C_1\cdot \max_{\card{S}\sim d_k}\card{\widehat{f}(S)\widehat{g}(S)} + C_2\cdot I[f,g]
    +C_3\cdot \norm{g}_2\norm{f}_2.
\end{equation}
$C_1 = 2^{k} d_k^{(1+\eps)d_k}\left(\frac{2}{\delta_k}\right)^{\frac{(1+\eps)d_k}{d_{k-1}}}$,
~
$C_2 = 2^{k}\delta_1^{\frac{1}{8}} 3^{\frac{d_1}{4}}$,
~
$C_3 = 2^{k}(1+\eps)^k d_k\sum_{j=1}^{k-1} d_j^{(1+\eps)d_j}\left(\frac{2}{\delta_j}\right)^{\frac{(1+\eps)d_j}{d_{j-1}}}\sqrt{3}^{d_{j+1}}\delta_{j+1}^{\quarter}$.
\end{thm}

The amount of parameters in the above statement, as well as the formulas for $C_1,C_2,C_3$, make the above statement incomprehensible;
this form is very convenient for the inductive proof to go through. Once it has been established, one can make a particular choice
of the parameters that is typically useful, yielding the following corollary.

\begin{corollary}\label{corr:main_set_param}
  Let $\alpha\in(0,1)$, $\eps\in(0,1/2)$
  and let $d\in\mathbb{N}$, $\delta>0$ be such that
  $d^{\eps}\geq \frac{100}{\alpha \eps}\log d$
  and $\delta\leq 2^{-4d^{\eps}}$.
  If $f\colon \power{n}\to\power{}$ is a Boolean function, and $g\colon\power{n}\to\R$ is $(\alpha,d)$-almost-homogenous,
  then
\[
  \inner{f}{g} \leq \delta^{-2^{O(1/\eps)} d}\cdot \max_{\card{S}\sim d}\card{\widehat{f}(S)\widehat{g}(S)}
                    + 2^{O(1/\eps)}\delta^{1/40}\cdot I[f,g]
    +2^{O(1/\eps)} \delta^{d^{\eps}}\cdot \norm{g}_2\norm{f}_2.
\]
\end{corollary}
\begin{proof}
  Assume $1/\eps$ is an integer (otherwise we may replace $\eps$ with some $\eps'$ such that $\eps\leq \eps'\leq 2\eps$ for which $1/\eps'$
  is an integer), and set $k = \frac{1}{\eps}$.
  Choose $d_j = d^{j/k}$ for all $i=0,1,\ldots,k$, and note that by the lower bound on $d$, we have
  $d_{j-1}\geq \frac{3}{\alpha\eps^2(1-2\eps)^{k}}\log(4(1+\eps)d_{j}/d_{j-1})$ for all $j\geq 3$.
  Next, we choose $\delta_1 = \delta$ and $\delta_{j+1} = \delta_j^{80 d^{\eps}}$ for all $j\geq 1$. We apply
  Theorem~\ref{thm:main_formal_fourier} with these parameters to upper bound $\inner{f}{g}$, and get that
  $\inner{f}{g}\leq C_1 \max_{\card{S}\sim d_k}\card{\widehat{f}(S)\widehat{g}(S)} + C_2\cdot I[f,g] +C_3\cdot \norm{g}_2\norm{f}_2$
  for $C_1,C_2,C_3$ as in the statement of Theorem~\ref{thm:main_formal_fourier}, and next we give simpler upper bounds on
  $C_1,C_2,C_3$ for our specific choice of parameters.

  Unraveling the definition of $\delta_{j+1}$, we see that $\delta_{j+1} = \delta^{(80 d^{\eps})^j}$, therefore
  $C_1 \leq d^{O(d)} \delta^{-{\sf exp}(O(1/\eps)) d}$, and since $\delta\leq 2^{-d^{\eps}}\leq 2^{-\log d} = 1/d$,
  we get that the $d^{O(d)}$ factor can be absorbed into the second factor.
  For $C_2$, we see that
  $C_2 \leq 2^{O(1/\eps)} \delta^{1/8} 3^{d^{\eps}/4}\leq 2^{O(1/\eps)} \delta^{1/40}$ using
  the upper bound on $\delta$. Finally, for $C_3$, consider the $j$th summand;
  note that $d_j^{d_j} \leq 2^{d_j \log d}\leq 2^{d_j d^{\eps}} = 2^{d_{j+1}}$,
  and that $\delta_{j+1}^{1/4}\leq \delta_j^{20 d^{\eps}}$. Thus,
  \[
  C_3
  \leq 2^k (1+\eps)^k d \sum\limits_{j=1}^{k-1} 12^{d_{j+1}} \delta_j^{16 d^{\eps}}
  \leq 2^k (1+\eps)^k d \sum\limits_{j=1}^{k-1} 12^{d^{(j+1)\eps}} \delta^{16 d^{j\eps}}
  \leq 2^k (1+\eps)^k d \sum\limits_{j=1}^{k-1} \delta^{8 d^{j\eps}},
  \]
  where in the last inequality we used the fact that $\delta\leq 2^{-d^{\eps}}$. We bound the sum by $k$ times
  the maximum summand and get that $C_3 \leq 2^{O(1/\eps)} d \delta^{8 d^{\eps}}\leq 2^{O(1/\eps)} \delta^{d^{\eps}}$.

\end{proof}

\subsection{Base case}
The base case $k=1$ of Theorem~\ref{thm:main_formal_fourier} is an easy consequence of Lemma~\ref{lem:base_case} below.
\begin{lemma}\label{lem:base_case}
Let $f\colon \power{n}\to \power{}$ be a Boolean function and let $g\colon\power{n}\to\R$ be
of degree at most $d$. Then for all $\delta>0$
we have that
\[
  \inner{f}{g}\leq \left(\frac{d}{\delta}\right)^d \max\card{\widehat{f}(S)\widehat{g}(S)}+2\cdot\delta^{1/8}3^{d/4} I[f,g].
\]
\end{lemma}
\begin{proof}
Note that without loss of generality, we may assume that for every $S$, the signs of the coefficients of $S$ in $f$ and $g$ are the same: indeed,
for any other $S$ we may change the sign of $\widehat{g}(S)$, leave the right hand side unchanged and only increase the left hand side
(as evident from Parseval/ Plancherel).

Denote by ${\sf Inf}[f]$ the set of variables $i$ such that $I_i[f^{\leq d}]\geq \delta$. Writing
$\inner{f}{g} = \sum\limits_{S}{\widehat{f}(S)\widehat{g}(S)}$, we partition the sum on the right hand side
to $S\subseteq {\sf Inf}[f]$ (i.e.\, only contain variables with high low-degree influence), and the rest.
Clearly, we have
\[
  \inner{f}{g}
  = \sum_{S}\widehat{f}(S)\widehat{g}(S)
  \leq \underbrace{\sum_{S\subseteq {\sf Inf}[f]}\widehat{f}(S)\widehat{g}(S)}_{(\rom{1})}+
  \sum_{i\not\in{\sf Inf}[f]}\underbrace{\sum_{S\ni i}\widehat{f}(S)\widehat{g}(S)}_{(\rom{2})}.
\]
Note that $\card{{\sf Inf}[f]}\leq \frac{d}{\delta}$, so the total number of summands in $(\rom{1})$ is
at most $\left(\frac{d}{\delta}\right)^d$, and we get that
\[
  (\rom{1})\leq \left(\frac{d}{\delta}\right)^d\max_S\widehat{f}(S)\widehat{g}(S).
\]

We now proceed to upper bound $(\rom{2})$ for each $i$ separately.
Fix $i\not\in {\sf Inf}[f]$. Since the sum is only supported on $S$ of size at most $d$, we may replace
$\widehat{f}(S)$ in that sum with $\widehat{f^{\leq d}}(S)$ and not change it. Hence,
we get that
\[
(\rom{2})= \inner{\partial_i f^{\leq d}}{\partial_i g},
\]
and using the Cauchy-Schwarz inequality $(\rom{2})$ is upper bounded by
$\norm{\partial_i f^{\leq d}}_2\norm{\partial_i g}_2$. We wish to upper bound the first multiplicand further,
and for that we note that $\norm{\partial_i f^{\leq d}}_2^2
= \inner{\partial_i f^{\leq d}}{\partial_i f^{\leq d}}
= \inner{\partial_i f^{\leq d}}{\partial_i f}$
and then use H\"{o}lder's inequality with powers $(4,4/3)$ to get that
\[
\inner{\partial_i f^{\leq d}}{\partial_i f}
\leq
\norm{\partial_i f^{\leq d}}_4
\norm{\partial_i f}_{4/3}
\leq
\sqrt{3}^d
\norm{\partial_i f^{\leq d}}_2
\norm{\partial_i f}_{4/3},
\]
where in the last inequality we used Theorem~\ref{thm:hypercontractivity}.
Since $i\not\in {\sf Inf}[f]$, we have that
$\norm{\partial_i f^{\leq d}}_2\leq \delta^{1/4} \cdot \norm{\partial_i f^{\leq d}}_2^{1/2}$,
which by Parseval is at most $\delta^{1/4}\cdot I_i[f]^{1/4}$.
To upper bound $\norm{\partial_i f}_{4/3}$, as $\partial_i f$ is $\set{0,1/2,-1/2}$-valued, we have that
$\norm{\partial_i f}_{4/3}\leq 2 I_i[f]^{3/4}$.
Combining the two bounds and taking square root, we get that
$\norm{\partial_i f^{\leq d}}_2\leq 2\cdot 3^{d/4}\delta^{1/8} I_i[f]^{1/2}$, and therefore
$
(\rom{2})\leq 2\cdot 3^{d/4} \delta^{\frac{1}{8}}\norm{\partial_i{f}}_2\norm{\partial_i{g}}_2
$.
\end{proof}

\subsection{Proof of inductive step}\label{sec:inductive_step}
In this section, we prove Theorem~\ref{thm:main_formal_fourier} by induction on $k$.
The  base case $k=1$ follows from Lemma~\ref{lem:base_case}, noting that since in our case $g$ is
$(\alpha,d_1)$-almost-homogenous, the maximum could be restricted to $\card{S}\sim d_1$.

Let $k>1$, assume the statement holds for all $j < k$
and prove it for $k$. To simplify notation, we recall that by $\card{S}\sim d_k$ we mean that
$\alpha d_k\leq \card{S}\leq d_k$; also, for $j\leq k-1$, we say that $\card{S}\sim d_j$ if
$(1-2\eps)\alpha d_j\leq \card{S}\leq d_j$.

\paragraph{Proof overview of the inductive step.} Let $f,g$ be functions as in the statement of the theorem,
and consider a partition $\mathcal{I} = (I_1,\ldots,I_m)$ into $m = d_k/d_{k-1}$ parts as in Lemma~\ref{lem:random_partition}; this
partition could be thought of as random.
We discard from $g$ characters $\chi_S$ for which $\card{S\cap I_j}\gg d_{k-1}$ for some $j\in[m]$.
Thus, thinking of $g$ as a function of only one of the parts, say $I_j$, its degree is at most $d_{k-1}$.
Our goal is to charge $S$'s that contribute to $\inner{f}{g} = \sum\limits_{S}{\widehat{f}(S)\widehat{g}(S)}$ to the various
random restrictions $f_{\bar{I_j}\rightarrow z}, g_{\bar{I_j}\rightarrow z}$ for $j\in[m]$, where $z$ is a random setting outside
the coordinates of $I_j$, and bound the contribution to the random restrictions using the induction hypothesis.

As hinted earlier, if all of the generalized influences of $g$ corresponding to subsets of $I_j$ are small,
then the induction hypothesis allows us to establish a good on bound $\inner{f}{g}$ by writing it as
$\Expect{z}{\inner{f_{\bar{I_j}\rightarrow z}}{g_{\bar{I_j}\rightarrow z}}}$.
To see that, we focus on the most problematic
term that arises from the induction hypothesis
involving the maximum over Fourier coefficients (the rest are significantly easier to handle), i.e.\
\[
\Expect{z\in\power{\bar{I_j}}}
{\max_S \card{\widehat{f_{\bar{I_j}\rightarrow z}}(S) \widehat{g_{\bar{I_j}\rightarrow z}}(S)}}.
\]
To prove a good upper bound on this term, one uses Cauchy-Schwarz and then Lemma~\ref{lem:expectation+maximus_exchange} to exchange
the maximum and expectation, which results in a bound depending on the generalized influences of $g$.

However, $g$ could of course have large generalized influences on $I_j$. Thus, we take $g_j$ to be the part of the Fourier
transform of $g$ that consists only of characters $\chi_S$ for which $S\cap I_j$ has small generalized influences (if for some $S$ we have that
$S\cap I_j$ is non-influential for more than a single $j$, we choose one such $j$ arbitrary and include $\chi_S$ in $g_j$).
Thus, we are able to upper bound $\inner{f}{g_j}$ successfully using the above strategy.
Subsequently we can decompose $g$ as $\sum\limits_{j\in [m]} g_j + E$,
where $E$ consists of characters $\chi_S$ for which $S\cap I_j$ is influential for all $j\in[m]$.
The task then amounts to upper bounding $\inner{f}{E}$, and for that we use the crude upper bound
$\card{{\sf supp}(\widehat{E})}\max_{S}\card{\widehat{f}(S)\widehat{E}(S)}$, which is sufficient
(using Fact~\ref{fact:sum_gen_inf} to upper bound the size of the support of $\widehat{E}$).

\skipi
We now move on to the formal proof.
\begin{proof}[Formal proof]
As before, we may assume without loss of generality that for all Fourier coefficients $S$, the signs of $\widehat{f}(S)$ and $\widehat{g}(S)$ are the same.

Using Corollary~\ref{cor:most_weight_on_good_charachters} with $m = (1+\eps)\frac{d_k}{d_{k-1}}$, $d=d_k$, $\alpha$ and $\eps$
we may find a partition $\mathcal{I}$ as in the corollary; let $G({\cal I})$ be the set of all $S\subseteq [n]$ such that for all $j\in[m]$,
$S\cap I_j$ is of size $(1+\eps)d_k/m = d_{k-1}$ within factor $(1-2\eps)\alpha$.

Define $\tilde{g}\colon\power{n}\to\R$ by
$\tilde{g}=\sum_{S\in G(\mathcal{I})}\widehat{g}(S)\chi_S$. In terms of $\tilde{g}$,
Corollary~\ref{cor:most_weight_on_good_charachters} amounts to saying that
\[
\inner{f}{\tilde{g}}\geq
\left(1-2me^{-\frac{\eps^2\alpha d}{3m}}\right)\inner{f}{g},
\]
and by the condition relating $d_k$ and $d_{k-1}$ we get that the factor
on the right hand side is at least $\half$, and therefore $\inner{f}{g}\leq 2\inner{f}{\tilde{g}}$
so it is enough to upper bound the inner product of $f$ and $\tilde{g}$.

Let $\mathcal{T}$ be the set of $T\subseteq[n]$ of size $d_{k-1}$ within factor $(1-\eps)\alpha$
that have large generalized influence in $\tilde{g}$,
i.e.\ such that $I_T[\tilde{g}]\geq \delta_k\norm{\tilde{g}}_2^2$.
Since $\tilde{g}$ has degree at most $d_k$, by Fact~\ref{fact:sum_gen_inf}
we get that
\[
\sum\limits_{\card{T}\sim d_{k-1}}{I_T[\tilde{g}]}
\leq \sum\limits_{\card{T}\leq d_{k-1}}I_T[\tilde{g}]
\leq 2 d_k^{d_{k-1}} \norm{g}_2^2,
\]
hence $\card{\mathcal{T}}\leq \frac{2d_k^{d_{k-1}}}{\delta_k}$. Writing $\inner{f}{\tilde{g}} = \sum\limits_{S}{\widehat{f}(S)\widehat{\tilde{g}}(S)}$,
we partition the sum on the right hand side into two parts:
$(\rom{1})$ those $S$ that satisfy that $S\cap I_j$ is in $\mathcal{T}$ for all $j\in[m]$,
and $(\rom{2})$ those $S$ such that $S\cap I_j$ is not in $\mathcal{T}$ for some
$j\in[m]$. Denote by $\mathcal{S}_j$ the set of $S$ such that $S\cap I_j\not\in\mathcal{T}$,
and by $\mathcal{B}$ the set of $S$ that are outside $\mathcal{S}_1\cup\ldots\cup\mathcal{S}_m$.
Then we have
\begin{equation}\label{eq:inductive_step_2}
\inner{f}{\tilde{g}}\leq
\underbrace{\sum_{S\in \mathcal{B}}\widehat{f}(S)\widehat{\tilde{g}}(S)}_{(\rom{1})}
+\underbrace{\sum_{j=1}^{m}\sum_{S\in{\cal S}_j}\widehat{f}(S)\widehat{\tilde{g}}(S)}_{(\rom{2})},
\end{equation}
and we upper bound each sum separately.
\paragraph{Upper bounding $(\rom{1})$.}
Clearly, $(\rom{1})$ is at most $\card{\mathcal{B}} \cdot \max_{\card{S}\sim d_k} \widehat{f}(S)\widehat{\tilde{g}}(S)$
(as the sum is only supported on $\card{S}\sim d_k$ by the condition on $g$).
To bound the size of $\mathcal{B}$, note that the map $S\rightarrow (S\cap I_1,\ldots,S\cap I_m)$
is a bijection from $\mathcal{B}$ to $\mathcal{T}^m$, hence we have that
\[
\card{\mathcal{B}}\leq \card{\mathcal{T}}^m\leq \left( \frac{2d_k^{d_{k-1}}}{\delta_k}\right)^m
=d_k^{(1+\eps)d_k}\left(\frac{2}{\delta_k}\right)^{\frac{(1+\eps) d_k}{d_{k-1}}}
\leq \half C_1(k).
\]
Therefore, the contribution from
$(\rom{1})$ is upper bounded by the first term on the right hand side in \eqref{eq:thm_statement}.
Thus, to complete the proof, it is enough to upper bound the contribution from $(\rom{2})$ by the other
two terms in the right hand side of \eqref{eq:thm_statement}, and to do so we use the induction hypothesis.

\paragraph{Upper bounding $(\rom{2})$.}
We upper bound the sum corresponding to each $j\in[m]$ separately. Fix $j$,
write $J = \cup_{j'\neq j} I_{j'}$, and $\tilde{g}_j = \sum\limits_{S\in\mathcal{S}_j}{\widehat{\tilde{g}}(S)\chi_S}$.
Note that
\[
\Expect{z}{\inner{f_{J\rightarrow z}}{(\tilde{g}_j)_{J\rightarrow z}}}
= \inner{f}{\tilde{g}_j}
= \sum_{S\in\mathcal{S}_j}\widehat{f}(S)\widehat{\tilde{g}}(S),
\]
and therefore to upper bound $(\rom{2})$ it is enough to upper bound the inner product of random restrictions of
$f$ and $\tilde{g}_j$ on $J$. We note that the important point here is that these restrictions lower the degree of
$\tilde{g}_j$ from $\sim d_k$ to $\sim d_{k-1}$ (since it is only supported on $S$ such that $\card{S\cap I_j}\sim d_{k-1}$),
and hence we expect to get useful information from the inductive hypothesis on these restrictions.

More precisely, note that for every $z\in\power{J}$ the function $f_{J\rightarrow z}$ is Boolean
and the function $(\tilde{g}_j)_{J\rightarrow z}$ is $((1-2\eps)\alpha,d_{k-1})$-approximately homogenous.
Therefore, we may apply the induction hypothesis with parameters $k-1$, $(1-2\eps)\alpha$, $\eps$, $d_0,\ldots,d_{k-1}$
and $\delta_{k-1},\ldots,\delta_{1}$ on these functions, to get that
\begin{align}
  \inner{f_{J\rightarrow z}}{(\tilde{g}_j)_{J\rightarrow z}}
  & \leq
  C_1(k-1)
  \underbrace{\max_{\card{S}\sim d_{k-1}}\card{\widehat{f_{J\rightarrow z}}(S)}\card{\widehat{(\tilde{g}_j)_{J\rightarrow z}}(S)}}_{(\rom{3})}
     +C_2(k-1)
  \underbrace{I[f_{J\rightarrow z},(\tilde{g}_j)_{J\rightarrow z}]}_{(\rom{4})} \notag\\
    & +C_3(k-1)
  \underbrace{\norm{(\tilde{g}_j)_{J\rightarrow z}}_2\norm{f_{J\rightarrow z}}_2}_{(\rom{5})}
  \label{eg:inductive_step}.
\end{align}
Here again, we denote
\begin{align*}
&C_1(r) = 2^{r}\cdot d_r^{(1+\eps)d_r}\left(\frac{2}{\delta_r}\right)^{\frac{(1+\eps)d_r}{d_{r-1}}},
\qquad\qquad C_2(r) = 2^{r}\cdot \delta_1^{\frac{1}{8}} 3^{\frac{d_1}{4}},\\
&\qquad
C_3(r) = 2^{r}(1+\eps)^r d_r\sum_{\ell=1}^{r-1} d_{\ell}^{(1+\eps)d_{\ell}}\left(\frac{2}{\delta_{\ell}}\right)^{\frac{(1+\eps)d_{\ell}}{d_{\ell-1}}}\sqrt{3}^{d_{\ell+1}}\delta_{\ell+1}^{\quarter},
\end{align*}
and we wish to upper bound the expectation of the left hand side.
We bound each one of them separately.

\paragraph{Upper bounding the contribution of $(\rom{5})$.}
By Cauchy-Schwarz we have that
\begin{align*}
\Expect{z}{(\rom{5})}
=\Expect{z}{\norm{(\tilde{g}_j)_{J\rightarrow z}}_2\norm{f_{J\rightarrow z}}_2}
\leq
\sqrt{\Expect{z}{\norm{(\tilde{g}_j)_{J\rightarrow z}}_2^2}\Expect{z}{\norm{f_{J\rightarrow z}}_2^2}}
= \norm{\tilde{g}_j}_2\norm{f}_2
\leq \norm{g}_2\norm{f}_2.
\end{align*}

\paragraph{Upper bounding the contribution of $(\rom{4})$.}
Using Lemma~\ref{lem:cross_inf}, we have that
\[
\Expect{z}{(\rom{4})}
=\Expect{z}{I[f_{J\rightarrow z},(\tilde{g}_j)_{J\rightarrow z}]}
\leq \sum\limits_{i\not\in J}{I_i[f,\tilde{g}_j]}
\leq \sum\limits_{i\in I_j}{I_i[f,g]},
\]
where we used the fact that
$I_i[\tilde{g}_j]\leq I_i[g]$ (that follows from Parseval's equality).

\paragraph{Upper bounding the contribution of $(\rom{3})$.} To upper bound the expectation of
$(\rom{3})$, we first use Cauchy-Schwarz inequality to see that
\[
\Expect{z}{(\rom{3})}
\leq\sqrt{
\Expect{z}{\max_{\card{S}\sim d_{k-1}}\widehat{f_{J\rightarrow z}}(S)^2}}
\sqrt{
\Expect{z}{\max_{\card{S}\sim d_{k-1}}\widehat{(\tilde{g}_j)_{J\rightarrow z}}(S)^2}
}.
\]
For each $z$ and each $\card{S}\sim d_{k-1}$, by Parseval we have that
$\widehat{f_{J\rightarrow z}}(S)^2\leq \norm{f_{J\rightarrow z}}_2^2$, hence,
the first multiplicand is bounded by $\norm{f}_2$.

For the second multiplicand we appeal to Lemma~\ref{lem:expectation+maximus_exchange}: let us think of the indices
therein as being subsets $S$, and define $h_S(z) = \widehat{(\tilde{g}_j)_{J\rightarrow z}}(S)$;
note that since $g$ is of degree at most $d_k$, we get that each $h_S$ also has degree at most $d_k$.
Thus, applying Lemma~\ref{lem:expectation+maximus_exchange} we get that
\[
\Expect{z}{\max_{S} h_S(z)^2}
\leq 3^{d_k}\max_S\norm{h_S}_2 \left(\Expect{z}{\sum\limits_{S} h_S(z)^2}\right)^{1/2}.
\]
Note that by Parseval,
\[
\Expect{z}{\sum\limits_{S} h_S(z)^2}
=\Expect{z}{\norm{(\tilde{g}_j)_{J\rightarrow z}}_2^2}
=\norm{\tilde{g}_j}_2^2
\leq \norm{g}_2^2,
\]
and therefore we conclude that
$\Expect{\alpha}{(\rom{3})}\leq \sqrt{3}^{d_k}\norm{g}_2^{1/2}\norm{f}_2\max_S\norm{h_S}_2^{1/2}$. Fix $S$ that attains this maximum;
using Lemma~\ref{lem:rr_coef}, we see that $\norm{h_S}_2^2 =\sum\limits_{T: T\cap I_j = S}{\widehat{\tilde{g}_j}(T)^2}$. Note
that if $S\in\mathcal{T}$, then the above sum would be empty (since we do not include a character $T$ in $\tilde{g}_j$ if
$T\cap I_j$ has large generalized influence in $g$), and the sum would be $0$. Hence we may assume that $S\not\in\mathcal{T}$,
and therefore this sum is at most $I_S[g]\leq \delta_k \norm{g}_2^2$. Combining, we get that
$\Expect{\alpha}{(\rom{3})}\leq \sqrt{3}^{d_k} \delta_k^{1/4}\norm{g}_2 \norm{f}_2$.

\paragraph{Combining the bounds for $(\rom{3}),(\rom{4}),(\rom{5})$.}
Plugging the bounds into ~\eqref{eg:inductive_step} and summing over $j\in[m]$, we get that
\[
(\rom{2})
=
\sum\limits_{j=1}^{m}\Expect{z}{\inner{f_{J\rightarrow z}}{(\tilde{g}_j)_{J\rightarrow z}}}
\leq
m\cdot C_1(k-1)\sqrt{3}^{d_k} \delta_k^{1/4}\norm{g}_2
+C_2(k-1)I[f,g]
+m\cdot C_3(k-1)\norm{g}_2\norm{f}_2.
\]
Consider the first and the third terms on the right hand side. Note that
\[
m\cdot C_1(k-1)\sqrt{3}^{d_k} \delta_k^{1/4} +
m\cdot C_3(k-1)
\leq \half C_3(k),
\]
as well as that $C_2(k-1) = \half C_2(k)$.
Therefore the above inequality implies that
$(\rom{2})\leq \half C_2(k)I[f,g] + \half C_3(k)\norm{g}_2\norm{f}_2$.
Plugging this, as well as the bound we have on $(\rom{1})$, into \eqref{eq:inductive_step_2}, we get that
\[
\inner{f}{\tilde{g}}
\leq
\half C_1(k) \max_{\card{S}\sim d_k} \widehat{f}(S)\widehat{\tilde{g}}(S)
+\half C_2(k)I[f,g] + \half C_3(k)\norm{g}_2\norm{f}_2,
\]
and since $\inner{f}{g}\leq 2\inner{f}{\tilde{g}}$, the proof of the inductive step is complete.
\end{proof}

\section{Improving the result}\label{sec:improved_results}
In this section, we prove quantitatively stronger forms of Theorem~\ref{thm:main_formal_fourier}
and Corollary~\ref{corr:main_set_param}. We will assume some familiarity with the material presented
in Sections~\ref{sec:rest,part,deg}, and~\ref{sec:main_pfs} and encourage the reader to go over
them before proceeding to the current section.

\subsection{A variant of Lemma~\ref{lem:expectation+maximus_exchange}}
We begin by proving the following variant of Lemma~\ref{lem:expectation+maximus_exchange},
which is the source of the improvement. The set-up one should have in mind
is the following: we have a Boolean function $f$ (whose degree is not necessarily small),
a low-degree function $g$, and a set of variables $J\subseteq[n]$. The functions $h_i$
are the Fourier coefficients of $g_{J\rightarrow x}$,
and the functions $h_i'$ are Fourier coefficients of $f_{J\rightarrow x}$.
I.e.\, we think of the indices $i$ as subsets $S\subseteq \bar{J}$, and have
$h_S(z) = \widehat{g_{J\rightarrow z}}(S)$, $h'_S(z) = \widehat{f_{J\rightarrow z}}(S)$.
We remark that $h_S$ play the same role as in Lemma~\ref{lem:expectation+maximus_exchange},
whereas $h'_S$ did not appear there.
\begin{lemma}\label{lem:expectation+maximus_exchange_improved}
 Let $d\in\mathbb{N}$.
 Let $h_1,h_2,\dots,h_k\colon \power{n}\to\R$ be of degree at most $d$,
 and $h_1',\ldots,h_k'\colon\power{n}\to \mathbb{R}$. Then
  \begin{align*}
  \Expect{x}{\left(\sum\limits_{i=1}^{k}\card{h_i'(x)}^{\frac{4}{3}} \card{h_i(x)}^{\frac{4}{3}}\right)^{\frac{3}{4}}}
  \leq
  3^{d} \max_i(\norm{h_i}_2\norm{h_i'}_2)^{\frac{1}{4}}
  \cdot\left(\Expect{x}{\sum\limits_{i=1}^{k} h_i'(x)^2}
  \Expect{x}{\sum\limits_{i=1}^{k} h_i(x)^2}\right)^{\frac{3}{8}}.
  \end{align*}
  \end{lemma}
  \begin{proof}
  By Jensen's inequality, the left hand side is at most
   \begin{equation}\label{eq1_improved}
   \left(\Expect{x}{\sum\limits_{i=1}^{k}\card{h_i'(x)}^{\frac{4}{3}} \card{h_i(x)}^{\frac{4}{3}}}\right)^{\frac{3}{4}}
   = \left(\sum\limits_{i=1}^{k}\Expect{x}{\card{h_i'(x)}^{\frac{4}{3}} \card{h_i(x)}^{\frac{4}{3}}}\right)^{\frac{3}{4}}.
   \end{equation}
 By H\"{o}lder's inequality, we have that
 \[
 \Expect{x}{\card{h_i'(x)}^{\frac{4}{3}} \card{h_i(x)}^{\frac{4}{3}}}
 \leq \Expect{x}{h_i'(x)^2}^{\frac{2}{3}} \Expect{x}{h_i(x)^4}^{\frac{1}{3}}
 = \norm{h_i'}_2^{\frac{4}{3}}\norm{h_i}_4^{\frac{4}{3}},
 \]
 Using Theorem~\ref{thm:hypercontractivity} we have
 $\norm{h_i}_4\leq 3^{d/2} \norm{h_i}_2$.
 Thus, we get that
 \[
 \Expect{x}{\card{h_i'(x)}^{\frac{4}{3}} \card{h_i(x)}^{\frac{4}{3}}}
 \leq
 3^d \norm{h_i}_2^{\frac{4}{3}}\norm{h_i'}_2^{\frac{4}{3}},
 \]
 and plugging this that into~\eqref{eq1_improved} yields that
 \[
 \eqref{eq1_improved}
 \leq 3^{d}\left(\sum_{i\in[k]}\norm{h_i}_2^{\frac{4}{3}}\norm{h_i'}_2^{\frac{4}{3}}\right)^{\frac{3}{4}}
 \leq 3^{d} \max_i(\norm{h_i}_2\norm{h_i'}_2)^{\frac{1}{4}}\left(\sum_{i\in[k]}\norm{h_i}_2\norm{h_i'}_2\right)^{\frac{3}{4}},
 \]
 The result now follows from the Cauchy-Schwarz inequality.
  \end{proof}

\subsection{The improved statement}
The following statement is an improved form of Theorem~\ref{thm:main_formal_fourier}.
Specifically the bound we have on $C_1$ and $C_3$ is much better.
Roughly speaking, the most costly terms in $C_3$ are $(2/\delta_j)^{(1+\eps)d_{j}/4d_{j-1}}$,
so to make it small one must choose $\delta_{j+1}\sim \delta_j^{(1+\eps)^3 d_j/d_{j-1}}$.
This should be compared to Theorem~\ref{thm:main_formal_fourier}, in which we are forced
to pick $\delta_{j+1}\sim \delta_j^{4 d_j/d_{j-1}}$ if we want $C_3$ to be small;
this factor of $4$ in the exponent quickly blows up and becomes exponential in $k$,
which forces us to take moderate $k$ since otherwise $C_1$ would be large.

\begin{thm}\label{thm:main_formal_fourier_improved}
Let $k\in\mathbb{N}$, $\alpha,\eps\in(0,1)$,
let $d_0, d_1,\ldots,d_{k}$ be an increasing sequence such that
$d_0 = 1$ and for each $2\leq j\leq k$ we have $d_{j-1}\geq \frac{3}{\alpha\eps^2(1-2\eps)^{k}}\log(4(1+\eps)d_{j}/d_{j-1})$,
and let $0<\delta_k\leq\ldots\leq \delta_1$. Then there are $C_1,C_2, C_3$ specified below, such that the following holds.

If $f\colon \power{n}\to\power{}$ is a Boolean function, and $g\colon\power{n}\to\R$ is $(\alpha,d_k)$-almost-homogenous,
then
\begin{equation}\label{eq:thm_statement_improved}
  \inner{f}{g} \leq C_1\cdot \left(\sum\limits_{\card{S}\sim d_k}\card{\widehat{f}(S)}^{\frac{4}{3}}\card{\widehat{g}(S)}^\frac{4}{3}\right)^{\frac{3}{4}} + C_2\cdot I[f,g]
    +C_3\cdot \norm{g}_2\norm{f}_2.
\end{equation}
For $k = 1$, we have
$C_1 = \left(\frac{d_1}{\delta_1}\right)^{\frac{1}{4} d_1}$,
$C_2 = 2\cdot\delta_1^{\frac{1}{8}} 3^{\frac{d_1}{4}}$
and $C_3 = 0$. For $k > 1$, we have
$C_1 = 2^{k} \left(\frac{2}{\delta_k}\right)^{\frac{1+\eps}{4}\frac{d_k}{d_{k-1}}}$,
$C_2 = 2^{k}\delta_1^{\frac{1}{8}} 3^{\frac{d_1}{4}}$,
and
$C_3 =
2^{k}(1+\eps)^k d_k
\left(\left(\frac{d_1}{\delta_1}\right)^{\frac{1}{4} d_1}\delta_2^{\frac{1}{4}} +
\sum_{j=2}^{k-1} \left(\frac{2}{\delta_j}\right)^{\frac{1+\eps}{4} \frac{d_j}{d_{j-1}}}3^{d_{j+1}}\delta_{j+1}^{\frac{1}{4}}\right)$.
\end{thm}

As before, the complicated-looking form of the statement stems from the inductive proof.
Later, we pick a convenient setting of the parameters in Corollary~\ref{corr:main_set_param_improved},
from which the improvement over Corollary~\ref{corr:main_set_param} is more apparent.

\subsubsection{Base case}
The base case $k=1$ of Theorem~\ref{thm:main_formal_fourier_improved} is an easy consequence of Lemma~\ref{lem:base_case_improved} below.
\begin{lemma}\label{lem:base_case_improved}
Let $f\colon \power{n}\to \power{}$ be a Boolean function and let $g\colon\power{n}\to\R$ be
of degree at most $d$. Then for all $\delta>0$
we have that
\[
  \inner{f}{g}\leq \left(\frac{d}{\delta}\right)^{\frac{1}{4}d}\left(\sum\limits_{S} \card{\widehat{f}(S)}^{\frac{4}{3}}\card{\widehat{g}(S)}^{\frac{4}{3}}\right)^{\frac{3}{4}}
  +2\cdot\delta^{1/8}3^{d/4} I[f,g].
\]
\end{lemma}
\begin{proof}
Note that without loss of generality, we may assume that for every $S$, the signs of the coefficients of $S$ in $f$ and $g$ are the same: indeed,
for any other $S$ we may change the sign of $\widehat{g}(S)$, leave the right hand side unchanged and only increase the left hand side
(as evident from Parseval/ Plancherel).

Denote by ${\sf Inf}[f]$ the set of variables $i$ such that $I_i[f^{\leq d}]\geq \delta$. Writing
$\inner{f}{g} = \sum\limits_{S}{\widehat{f}(S)\widehat{g}(S)}$, we partition the sum on the right hand side
to $S\subseteq {\sf Inf}[f]$ (i.e.\, only contain variables with high low-degree influence), and the rest.
Clearly, we have
\[
  \inner{f}{g}
  = \sum_{S}\widehat{f}(S)\widehat{g}(S)
  \leq \underbrace{\sum_{S\subseteq {\sf Inf}[f]}\widehat{f}(S)\widehat{g}(S)}_{(\rom{1})}+
  \sum_{i\not\in{\sf Inf}[f]}\underbrace{\sum_{S\ni i}\widehat{f}(S)\widehat{g}(S)}_{(\rom{2})}.
\]
Note that $\card{{\sf Inf}[f]}\leq \frac{d}{\delta}$, so the total number of summands in $(\rom{1})$ is
at most $\left(\frac{d}{\delta}\right)^d$, and we get by H\"{o}lder's inequality that
\[
  (\rom{1})\leq \left(\frac{d}{\delta}\right)^{\frac{d}{4}}\left(\sum\limits_{S} \card{\widehat{f}(S)}^{\frac{4}{3}}\card{\widehat{g}(S)}^{\frac{4}{3}}\right)^{\frac{3}{4}}.
\]
The bound on $(\rom{2})$ is identical to the bound on $(\rom{2})$
in the proof of Lemma~\ref{lem:base_case}.
\end{proof}

\subsubsection{Proof of inductive step}\label{sec:inductive_step_improved}
Let $k>1$, assume the statement holds for all $j < k$
and prove it for $k$. To simplify notation, we recall that by $\card{S}\sim d_k$ we mean that
$\alpha d_k\leq \card{S}\leq d_k$; also, for $j\leq k-1$, we say that $\card{S}\sim d_j$ if
$(1-2\eps)\alpha d_j\leq \card{S}\leq d_j$.

\skipi

The proof below is the same as the proof in Section~\ref{sec:inductive_step} with two modifications:
we sieve out the $g$ to get the functions $\tilde{g}_j$ more carefully so as to make sure the generalized
influences of $g$ as well as of $f$ arising in the argument would be small, and use Lemma~\ref{lem:expectation+maximus_exchange_improved}
instead of Lemma~\ref{lem:expectation+maximus_exchange}.
\begin{proof}
As before, by flipping signs of the Fourier coefficients of $g$ if necessary,
we may assume without loss of generality that for all $S$, the signs of $\widehat{f}(S)$ and $\widehat{g}(S)$ are the same.

Using Corollary~\ref{cor:most_weight_on_good_charachters} with $m = (1+\eps)\frac{d_k}{d_{k-1}}$, $d=d_k$, $\alpha$ and $\eps$
we may find a partition $\mathcal{I}$ as in the corollary; let $G({\cal I})$ be the set of all $S\subseteq [n]$ such that for all $j\in[m]$,
$S\cap I_j$ is of size $(1+\eps)d_k/m = d_{k-1}$ within factor $(1-2\eps)\alpha$.

Define $\tilde{g}\colon\power{n}\to\R$ by
$\tilde{g}=\sum_{S\in G(\mathcal{I})}\widehat{g}(S)\chi_S$. In terms of $\tilde{g}$,
Corollary~\ref{cor:most_weight_on_good_charachters} amounts to saying that
\[
\inner{f}{\tilde{g}}\geq
\left(1-2me^{-\frac{\eps^2\alpha d}{3m}}\right)\inner{f}{g},
\]
and by the condition relating $d_k$ and $d_{k-1}$ we get that the factor
on the right hand side is at least $\half$, and therefore $\inner{f}{g}\leq 2\inner{f}{\tilde{g}}$
so it is enough to upper bound the inner product of $f$ and $\tilde{g}$.

For each $j\in[m]$ and $T\subseteq I_j$, denote $I_{T,I_j}[\tilde{g}] \defeq \sum\limits_{S: S\cap I_J = T}{\widehat{\tilde{g}}(S)^2}$,
and similarly define $I_{T,I_j}[f]\defeq \sum\limits_{S: S\cap I_J = T}{\widehat{f}(S)^2}$. Let $\mathcal{T}_j$ be the collection of all sets $T\subseteq I_j$ such
either $I_{T,I_j}[f]\geq \delta_k \norm{f}_2^2$ or $I_{T,I_j}[\tilde{g}]\geq \delta_k \norm{\tilde{g}}_2^2$.
Note that
\[
\sum\limits_{T\subseteq I_j}{I_{T,I_j}[\tilde{g}]}
= \sum\limits_{S}\tilde{g}^2(S)
= \norm{\tilde{g}}_2^2,
\]
hence $\tilde{g}$ contributes at most $1/\delta_k$ sets $T$ to $\mathcal{T}_j$; similarly,
$f$ also contributes at most $1/\delta_k$ sets $T$ to $\mathcal{T}_j$, and hence
$\card{\mathcal{T}_j}\leq \frac{2}{\delta_k}$.
Writing $\inner{f}{\tilde{g}} = \sum\limits_{S}{\widehat{f}(S)\widehat{\tilde{g}}(S)}$,
we partition the sum on the right hand side into two parts:
$(\rom{1})$ those $S$ that satisfy that $S\cap I_j$ is in $\mathcal{T}_j$ for all $j\in[m]$,
and $(\rom{2})$ those $S$ such that $S\cap I_j$ is not in $\mathcal{T}_j$ for some
$j\in[m]$. Denote by $\mathcal{S}_j$ the set of $S$ such that $S\cap I_j\not\in\mathcal{T}_j$,
and by $\mathcal{B}$ the set of $S$ that are outside $\mathcal{S}_1\cup\ldots\cup\mathcal{S}_m$.
Then we have
\begin{equation}\label{eq:inductive_step_2_improved}
\inner{f}{\tilde{g}}\leq
\underbrace{\sum_{S\in \mathcal{B}}\widehat{f}(S)\widehat{\tilde{g}}(S)}_{(\rom{1})}
+\underbrace{\sum_{j=1}^{m}\sum_{S\in{\cal S}_j}\widehat{f}(S)\widehat{\tilde{g}}(S)}_{(\rom{2})},
\end{equation}
and we upper bound each sum separately.
\paragraph{Upper bounding $(\rom{1})$.}
Note that since $g$ is almost homogenous, we get that $(\rom{1})$ is
only supported on $\card{S}\sim d_k$.
Therefore, using H\"{o}lder's inequality, $(\rom{1})$ is at most
$\card{\mathcal{B}}^{\frac{1}{4}} \cdot \left(\sum_{\card{S}\sim d_k} \card{\widehat{f}(S)}^{\frac{4}{3}}\card{\widehat{\tilde{g}}(S)}^{\frac{4}{3}}\right)^{\frac{3}{4}}$.
To bound the size of $\mathcal{B}$, note that the map $S\rightarrow (S\cap I_1,\ldots,S\cap I_m)$
is a bijection from $\mathcal{B}$ to $\mathcal{T}_1\times\ldots\times\mathcal{T}_m$, hence we have
that
\[
\card{\mathcal{B}}^{\frac{1}{4}}\leq
\left(\prod\limits_{i=1}^{m}\card{\mathcal{T}_i}\right)^{\frac{1}{4}}
=\left(\frac{2}{\delta_k}\right)^{m\cdot\frac{1}{4}}
=\left(\frac{2}{\delta_k}\right)^{\frac{1+\eps}{4} \frac{d_k}{d_{k-1}}}
\leq \frac{1}{2} C_1(k).
\]
Therefore, the contribution from
$(\rom{1})$ is upper bounded by the first term on the right hand side in \eqref{eq:thm_statement_improved}.
Thus, to complete the proof, it is enough to upper bound the contribution from $(\rom{2})$ by the other
two terms in the right hand side of \eqref{eq:thm_statement_improved}, and to do so we use the induction hypothesis.

\paragraph{Upper bounding $(\rom{2})$.}
We upper bound the sum corresponding to each $j\in[m]$ separately. Fix $j$,
write $J = \cup_{j'\neq j} I_{j'}$, and $\tilde{g}_j = \sum\limits_{S\in\mathcal{S}_j}{\widehat{\tilde{g}}(S)\chi_S}$.
Note that
\[
\Expect{z}{\inner{f_{J\rightarrow z}}{(\tilde{g}_j)_{J\rightarrow z}}}
= \inner{f}{\tilde{g}_j}
= \sum_{S\in\mathcal{S}_j}\widehat{f}(S)\widehat{\tilde{g}}(S),
\]
and therefore to upper bound $(\rom{2})$ it is enough to upper bound the inner product of random restrictions of
$f$ and $\tilde{g}_j$ on $J$. We note that the important point here is that these restrictions lower the degree of
$\tilde{g}_j$ from $\sim d_k$ to $\sim d_{k-1}$ (since it is only supported on $S$ such that $\card{S\cap I_j}\sim d_{k-1}$),
and hence we expect to get useful information from the inductive hypothesis on these restrictions.

More precisely, note that for every $z\in\power{J}$ the function $f_{J\rightarrow z}$ is Boolean
and the function $(\tilde{g}_j)_{J\rightarrow z}$ is $((1-2\eps)\alpha,d_{k-1})$-approximately homogenous.
Therefore, we may apply the induction hypothesis with parameters $k-1$, $(1-2\eps)\alpha$, $\eps$, $d_0,\ldots,d_{k-1}$
and $\delta_{k-1},\ldots,\delta_{1}$ on these functions, to get that
\begin{align}
  \inner{f_{J\rightarrow z}}{(\tilde{g}_j)_{J\rightarrow z}}
  & \leq
  C_1(k-1)
  \underbrace{
  \left(\sum\limits_{\card{S}\sim d_{k-1}}\card{\widehat{f_{J\rightarrow z}}(S)}^{\frac{4}{3}} \card{\widehat{(\tilde{g}_j)_{J\rightarrow z}}(S)}^{\frac{4}{3}}\right)^{\frac{3}{4}}}_{(\rom{3})}
     +C_2(k-1)
  \underbrace{I[f_{J\rightarrow z},(\tilde{g}_j)_{J\rightarrow z}]}_{(\rom{4})} \notag\\
    & +C_3(k-1)
  \underbrace{\norm{(\tilde{g}_j)_{J\rightarrow z}}_2\norm{f_{J\rightarrow z}}_2}_{(\rom{5})}
  \label{eg:inductive_step_improved}.
\end{align}
Here again, we denote $C_1(k-1), C_2(k-1), C_3(k-1)$ the expressions from Theorem~\ref{eq:thm_statement_improved}
when $k$ is replaced by $k-1$,
and we wish to upper bound the expectation of the left hand side.
We bound the expectation of each term separately.

\paragraph{Upper bounding the contribution of $(\rom{5})$.}
By Cauchy-Schwarz we have
\[
\Expect{z}{(\rom{5})}
=\Expect{z}{\norm{(\tilde{g}_j)_{J\rightarrow z}}_2\norm{f_{J\rightarrow z}}_2}
\leq\sqrt{\Expect{z}{\norm{(\tilde{g}_j)_{J\rightarrow z}}_2^2}\Expect{z}{\norm{f_{J\rightarrow z}}_2^2}},
\]
which is equal to $\norm{\tilde{g}}_2\norm{f}_2\leq \norm{g}_2\norm{f}_2$.

\paragraph{Upper bounding the contribution of $(\rom{4})$.}
Using Lemma~\ref{lem:cross_inf}, we have that
\[
\Expect{z}{(\rom{4})}
=\Expect{z}{I[f_{J\rightarrow z},(\tilde{g}_j)_{J\rightarrow z}]}
\leq \sum\limits_{i\not\in J}{I_i[f,\tilde{g}_j]}
\leq \sum\limits_{i\in I_j}{I_i[f,g]},
\]
where we used the fact that
$I_i[\tilde{g}_j]\leq I_i[g]$ (that follows from Parseval's equality).

\paragraph{Upper bounding the contribution of $(\rom{3})$.}
We use Lemma~\ref{lem:expectation+maximus_exchange_improved} to upper bound
the expectation of $(\rom{3})$. Let us think of the indices
therein as being subsets $S$, and define $h_S'(z) = \widehat{f_{J\rightarrow z}}(S)$
as well as $h_S(z) = \widehat{(\tilde{g}_j)_{J\rightarrow z}}(S)$. Since $\tilde{g}_j$ is of degree at most $d_k$,
we get that $h_S$ is of degree at most $d_k$.

By Lemma~\ref{lem:expectation+maximus_exchange_improved} we get that
\[
\Expect{z}{
\left(\sum\limits_{\card{S}\sim d_{k-1}}\card{h_S(z)}^{\frac{4}{3}}\card{h_S'(z)}^{\frac{4}{3}}\right)^{\frac{3}{4}}}
\leq
3^{d_k} \max_{\card{S}\sim d_{k-1}}\left(\norm{h_S}_2\norm{h_S'}_2\right)^{\frac{1}{4}}
\cdot\left(\Expect{z}{\sum\limits_{S} h_S'(z)^2}
  \Expect{z}{\sum\limits_{S} h_S(z)^2}\right)^{\frac{3}{8}}.
\]
Note that by Parseval,
\[
\Expect{z}{\sum\limits_{S} h_S(z)^2}
=\Expect{z}{\norm{(\tilde{g}_j)_{J\rightarrow z}}_2^2}
=\norm{\tilde{g}_j}_2^2
\leq \norm{g}_2^2,
\]
and similarly $\Expect{z}{\sum\limits_{S} h_S'(z)^2}\leq \norm{f}_2^2$.
Therefore we conclude that
\[
\Expect{z}{(\rom{3})}\leq
3^{d_k}\max_{\card{S}\sim d_{k-1}}\left(\norm{h_S}_2\norm{h_S'}_2\right)^{\frac{1}{4}}
\left(\norm{f}_2\norm{g}_2\right)^{\frac{3}{4}}.
\]
Fix $S$ that attains this maximum;
using Lemma~\ref{lem:rr_coef}, we see that
$\norm{h_S}_2^2 = \sum\limits_{T: T\cap I_j = S}{\widehat{\tilde{g}_j}(T)^2}$, which
is the same as $I_{S,I_j}[\tilde{g}]$. Note
that if $S\in\mathcal{T}_j$, then the above sum would be empty (since we do not include a character $T$ in $\tilde{g}_j$ if
$T\cap I_j\in \mathcal{T}_j$), and therefore its value would be $0$. Hence we may assume that $S\not\in\mathcal{T}_j$,
and thus we get that $I_{S,I_j}[g]\leq \delta_k \norm{g}_2^2$
and $I_{S,I_j}[f]\leq \delta_k \norm{f}_2^2$. Combining, we get that
\[
\Expect{z}{(\rom{3})}\leq 3^{d_k} \delta_k^{\frac{1}{4}}\norm{g}_2 \norm{f}_2.
\]

\paragraph{Combining the bounds for $(\rom{3}),(\rom{4}),(\rom{5})$.}
Plugging the bounds into~\eqref{eg:inductive_step_improved} and summing over $j\in[m]$, we get that
\begin{align}
(\rom{2})
=
\sum\limits_{j=1}^{m}\Expect{z}{\inner{f_{J\rightarrow z}}{(\tilde{g}_j)_{J\rightarrow z}}}
\leq
&m\cdot C_1(k-1) 3^{d_k} \delta_k^{\frac{1}{4}} \norm{g}_2 \norm{f}_2 \notag\\\label{eq:midproof}
&+C_2(k-1)I[f,g]
+m\cdot C_3(k-1)\norm{g}_2\norm{f}_2.
\end{align}
Consider the first and the third terms on the right hand side. Note that
\[
m\cdot C_1(k-1) 3^{d_k} \delta_k^{\frac{1}{4}} +
m\cdot C_3(k-1)
\leq \half C_3(k),
\]
as well as that $C_2(k-1) = \half C_2(k)$.
Therefore the above inequality implies that
\[
(\rom{2})\leq \half C_2(k)I[f,g] + \half C_3(k)\norm{g}_2\norm{f}_2.
\]
Plugging this, as well as the bound we have on $(\rom{1})$, into \eqref{eq:inductive_step_2_improved}, we get that
\[
\inner{f}{\tilde{g}}
\leq
\half C_1(k) \left(\sum_{\card{S}\sim d_k} \card{\widehat{f}(S)}^{\frac{4}{3}}\card{\widehat{\tilde{g}}(S)}^{\frac{4}{3}}\right)^{\frac{3}{4}}
+\half C_2(k)I[f,g] + \half C_3(k)\norm{g}_2\norm{f}_2,
\]
and since $\inner{f}{g}\leq 2\inner{f}{\tilde{g}}$, the proof of the inductive step is complete.
\end{proof}

%

\subsection{A convenient setting of the parameters}\label{sec:set_improved}
\begin{corollary}\label{corr:main_set_param_improved}
  Let $\alpha\in(0,1)$, $\eps\in(0,1/2)$
  and let $d\in\mathbb{N}$, $\delta>0$ be such that
  $d^{\eps}\geq \frac{100}{\alpha \eps}\log d$
  and $\delta\leq 2^{-\frac{16}{\eps} d^{\eps}}$.
  If $f\colon \power{n}\to\power{}$ is a Boolean function, and $g\colon\power{n}\to\R$ is $(\alpha,d)$-almost-homogenous,
  then
\[
  \inner{f}{g} \leq
  \delta^{-6 d}\cdot \norm{f}_2^{\frac{3}{4}}\norm{g}_2^{\frac{3}{4}}\cdot \max_{\card{S}\sim d}\card{\widehat{f}(S)\widehat{g}(S)}^{\frac{1}{4}}
  +\delta^{1/16}\cdot I[f,g]
  +\delta \cdot \norm{f}_2\norm{g}_2.
\]
\end{corollary}
\begin{proof}
  Assume $1/\eps$ is an integer (otherwise we may replace $\eps$ with some $\eps'$ such that $\eps\leq \eps'\leq 2\eps$ for which $1/\eps'$
  is an integer), and set $k = \frac{1}{\eps}$.
  Choose $d_j = d^{j/k}$ for all $j=0,1,\ldots,k$, and note that by the lower bound on $d$, we have
  $d_{j-1}\geq \frac{3}{\alpha\eps^2(1-2\eps)^{k}}\log(4(1+\eps)d_{j}/d_{j-1})$ for all $j\geq 2$.
  Next, we choose $\delta_1 = \delta$ and $\delta_{j+1} = \delta_j^{(1+\eps)^2 d^{\eps}}$ for all $j\geq 1$.
  We apply Theorem~\ref{thm:main_formal_fourier_improved} with these parameters to upper bound $\inner{f}{g}$, and get that
  \[
  \inner{f}{g}\leq C_1 \left(\sum\limits_{\card{S}\sim d_k}\card{\widehat{f}(S)}^{\frac{4}{3}}\card{\widehat{g}(S)}^{\frac{4}{3}}\right)^{\frac{3}{4}} + C_2\cdot I[f,g] +C_3\cdot \norm{g}_2\norm{f}_2
  \]
  for $C_1,C_2,C_3$ as in the statement of Theorem~\ref{thm:main_formal_fourier_improved}. For the first term observe that
  \[
  \sum\limits_{\card{S}\sim d_k}\card{\widehat{f}(S) \widehat{g}(S)}^{\frac{4}{3}}
  \leq \max_{\card{S}\sim d}\card{\widehat{f}(S)\widehat{g}(S)}^{\frac{1}{3}}
  \sum\limits_{\card{S}\sim d}\card{\widehat{f}(S)\widehat{g}(S)}
  \leq \max_{\card{S}\sim d}\card{\widehat{f}(S)\widehat{g}(S)}^{\frac{1}{3}}
  \sqrt{\sum\limits_{S}\widehat{f}(S)^2\sum\limits_{S}\widehat{g}(S)^2},
  \]
  where the last inequality is by Cauchy-Schwarz. By Parseval, this
  is equal to $\max\limits_{\card{S}\sim d}\card{\widehat{f}(S)\widehat{g}(S)}^{\frac{1}{3}}\norm{f}_2\norm{g}_2$.
  Next, we give simpler upper bounds on
  $C_1,C_2,C_3$ for our specific choice of parameters.

  Unraveling the definition of $\delta_{j+1}$, we see that $\delta_{j+1} = \delta^{(1+\eps)^{2j} d^{j\eps}}$, therefore
  $C_1  \leq 2^k \cdot 2^{d^{\eps}}\cdot \delta^{-\half(1+\eps)^{2k} d}$, and since $k = \frac{1}{\eps}$,
  we get that the $(1+\eps)^{2k}\leq e^2\leq 9$; using that and the fact that $2^{d^{\eps}}\leq \delta^{-d/2}$
  we conclude that $C_1\leq 2^k \delta^{-5 d}\leq \delta^{-6 d}$, where in
  the last inequality we used $\delta^{-1}\geq 2^{d^{\eps}}\geq 2^{1/\eps}=2^k$.

  \begin{claim}
    $C_3\leq  \delta$.
  \end{claim}
  \begin{proof}
  We upper bound each one of the summands in the definition of $C_3$ separately. Fix $2\leq j\leq k-1$;
  by the definition of $\delta_{j+1}$ we have
  have that the $j$th summand is
  \[
  \left(\frac{2}{\delta_j}\right)^{\frac{1+\eps}{4}d^{\eps}} 3^{d_{j+1}}\delta_{j+1}^{\frac{1}{4}}
  =  \left(\frac{2}{\delta_j}\right)^{\frac{1+\eps}{4} d^{\eps}}3^{d_{j+1}}\delta_{j}^{\frac{(1+\eps)^2}{4}d^{\eps}}
  \leq 6^{d_{j+1}}\delta_j^{\frac{\eps}{4} d^{\eps}}
  \leq 6^{d_{j+1}}\delta^{\frac{\eps}{4} d^{j\eps}}.
  \]
  Since $d_{j+1} = d^{(j+1)\eps}$ and by the upper bound we have on $\delta$, we get that this is at most $\delta^{\frac{\eps}{16} d^{j\eps}}\leq \delta^{2}$,
  where we used $\eps d^{j\eps}\geq \eps d^{\eps}\geq 100$.

  For $j=1$, the corresponding summand is
  \[
  \left(\frac{d^{\eps}}{\delta_1}\right)^{\frac{1}{4} d^{\eps}} \delta_2^{\frac{1}{4}}
  =\frac{d^{\frac{1}{4} \eps d^{\eps}}}{\delta_1^{\frac{1}{4} d^{\eps}}} \delta_1^{\frac{(1+\eps)^2}{4}d^{\eps}}
  \leq d^{\half \eps d^{\eps}} \delta^{\half \eps d^{\eps}}
  = (d\delta)^{\half\eps d^{\eps}}
  \leq \delta^{\quarter \eps d^{\eps}}
  \leq \delta^{2}.
  \]
  In the fourth transition, we used the upper bound on
  $\delta$ and the lower bound on $d^{\eps}$ that imply
  $\delta\leq 2^{-d^{\eps}}\leq 2^{-100 \log d} = d^{-100}$,
  and in particular $d\delta\leq \delta^{1/2}$. In the
  last transition we used $\eps d^{\eps}\geq 100$.

  Combining the bounds for all $j$'s we get that $C_3\leq 2^{k}(1+\eps)^k d\cdot k\cdot\delta^{2}$.
  Now, to see that this is at most $\delta$ simply note that $(1+\eps)^k\leq e$ and
  $\delta\leq 2^{-8 d^{\eps}}\leq  2^{-8 \log d/\eps}\leq \frac{\eps 2^{-1/\eps}}{e\cdot d} = \frac{2^{-k}}{k\cdot e\cdot d}$.
  \end{proof}

  Finally, for $C_2$ note that $3^{d_1/4} \leq 2^{d^{\eps}/2}$,
  $2^{k}\leq 2^{d^{\eps}}$ and $\delta_1^{1/16}\leq 2^{-d^{\eps}/\eps} \leq 2^{-2 d^{\eps}}$,
  and so $C_2\leq \delta^{1/16}$.
\end{proof}

\section{A further improvement of the result}\label{sec:improved_results2}
Corollary~\ref{corr:main_set_param_improved} from the last section can be used to established a quantitatively weaker
forms of Theorems~\ref{thm:main},~\ref{thm:main_gen} and~\ref{thm:main_entropy_improved}, where the logarithmic factor
$\log(1+\tilde{I}[f])$ is to be replaced by ${\sf poly}(\log(1+\tilde{I}[f]))$ (the proof is similar to the proofs we present
in Section~\ref{sec:imp_corollaries}). In this section, we first explain the reason for this quantitative inefficiency, and
then show how to improve on it. The overall idea and technique and very similar to what we have done so far, and we will focus on the places
in which the arguments differ.

Let us inspect Corollary~\ref{corr:main_set_param_improved} closely. First, we note that a good way to use this corollary is for example to
show that the total mass of $f$ on a given level contributed from small Fourier coefficients, e.g.\
\[
\sum\limits_{\card{S} = d}{\widehat{f}(S)^2 1_{\card{\widehat{f}(S)}\leq \xi}}
\]
for some appropriate $\xi$, is small. To do that, one defines $g = \sum\limits_{\card{S} = d}{\widehat{f}(S)1_{\card{\widehat{f}(S)}\leq \xi} \chi_S}$,
and notes that the above mass is simply $\inner{f}{g}$. The goal here to pick as large $\xi$ as possible, while still being to prove that $\inner{f}{g}$
is small (this type of problem naturally arises when one is trying to upper bound entropies as we will see in Section~\ref{sec:imp_corollaries}).

We now look more specifically at the conditions of Corollary~\ref{corr:main_set_param_improved} that are assumed for the parameters
(these turn out to be the cause of the quantitative inefficiency). First, note the factor $\delta^{-d}$ that multiplies the Fourier coefficients;
for the right hand side to be small, the magnitude of the Fourier coefficients must be at least as small as $\delta^d$ to cancel
that factor, hence to be most effective, we should try to pick as large $\delta$ as possible.

One condition in Corollary~\ref{corr:main_set_param_improved} asserts that we must have that $\eps d^{\eps}\geq \log d$,
which implies in particular that $\eps\geq \frac{2\log\log d-\log\log\log d-1}{\log d}$.
The condition on $\delta$ implies that we must have $\delta\leq 2^{-d^{\eps}/\eps}$. and thinking of $\frac{d^{\eps}}{\eps}$ as a function of
$\eps$, we get that it is increasing on the interval $\eps\in[1/\log d,\infty)$, and by the lower bound we argued for $\eps$
it follows that $d^{\eps}/\eps\geq \Omega(\log^3 d/(\log\log d)^2)$ and hence $\delta\leq 2^{-\Omega(\log^3 d/ (\log\log d)^2)}$. In this case,
the $\delta^{-d}$ factor on the right hand side of Corollary~\ref{corr:main_set_param_improved} is at least $2^{\Omega(d\log^3 d/ (\log\log d)^2)}$,
and for the Corollary to be useful the Fourier coefficients must be at most inversely small.

In other words, we cannot really say anything about Fourier coefficients of
$f$ that are of magnitude more than $2^{-d\log^3 d/(\log\log d)^2}$ except that the sum of their squares is $1$.
In principle, they alone may already contribute $d\frac{\log^3 d}{(\log\log d)^2}$ entropy, and hence to prove better
bounds on the entropy we must be able to strengthen the argument so that it can also handle them.

Towards this end, we first highlight the source of this bottleneck lies in the proof of Theorem~\ref{thm:main_formal_fourier_improved}.
Roughly speaking, the argument therein does successive degree reduction, which is successful as long as we are not trying to reduce
the degree to be too small (this is the reason behind the various requirements on the parameters). When our degree becomes too small,
we appeal to the base case, i.e.\ Lemma~\ref{lem:base_case_improved}, and this point turns out to be at degree ${\sf poly}(\log d)$.

To improve upon this argument, we must be able to reduce the degree down further (and not use the base case directly),
and this is the key in the quantitative improvement presented in this section. We show a slightly different way to analyze the degree reduction which on the one hand
is more efficient, but on the other hand incurs a more serious multiplicative error term for each iteration.
However, since we invoke such iterations only when the degree is already very small, there will be very few of them (roughly $\log\log d$)
and hence these additional error terms will be negligible. We also note that we need decrease the degree at each
step only by a constant factor as opposed to the setting described above, but this doesn't cause any additional complications.
\subsection{Reducing from small degrees to very small degrees using random partitions}
\begin{fact}\label{fact:binom_approx}
  Let $0\leq v\leq d$. Then
  \[
  {d\choose v}\geq \frac{1}{16\sqrt{v}}\frac{d^d}{(d-v)^{d-v}v^v}
  \]
\end{fact}
\begin{proof}
  By Stirling's approximation, we have that
  $\sqrt{n}\left(\frac{n}{e}\right)^n\leq n!\leq 4\sqrt{n}\left(\frac{n}{e}\right)^n$,
  and therefore
  \[
  {d\choose v}=
  \frac{d!}{v!(d-v)!}
  \geq \frac{\sqrt{d}\left(\frac{d}{e}\right)^d}{16 \sqrt{d-v}\left(\frac{d-v}{e}\right)^{d-v}\sqrt{v}\left(\frac{v}{e}\right)^v}
  \geq\frac{1}{16\sqrt{v}}\frac{d^d}{(d-v)^{d-v}v^v}.
  \]
\end{proof}
\begin{fact}
  Let $v<d\leq n$ and suppose $v$ divides $d$. Let $\mathcal{I} = (I_1,\ldots,I_{d/v})$ be a random partition into
  $d/v$ parts. Then for any $S\subseteq[n]$ of size $d$ it holds that
  \[
  \Prob{\mathcal{I}}{\card{S\cap I_i} = v~\forall i}\geq \left(\frac{1}{16\sqrt{v}}\right)^{d/v}.
  \]
\end{fact}
\begin{proof}
  The proof is by induction on $d/v$. For $d=v$ the claim is obvious. Fix $v$, assume the claim holds for all
  $d'<d$ and prove it for $d$. Note that
  \[
  \Prob{\mathcal{I}}{\card{S\cap I_i} = v~\forall i}
  = \Prob{I_{d/v}}{\card{S\cap I_{d/v}} = v}
  \cProb{\mathcal{I}}{\card{S\cap I_{d/v}} = v}{\card{S\cap I_i} = v~\forall i}.
  \]
  For the first term, by direct computation we have that
  \[
  \Prob{I_{d/v}}{\card{S\cap I_{d/v}} = v}
  =
  {d \choose v}\left(1-\frac{1}{d/v}\right)^{d-v}\left(\frac{1}{d/v}\right)^v
  ={d\choose v} \frac{(d-v)^{d-v}v^v}{d^d},
  \]
  and using Fact~\ref{fact:binom_approx} this is at least $\frac{1}{16\sqrt{v}}$. For the second term, conditioning on $I_{d/v}$ we have
  that
  \[
   \cProb{\mathcal{I}}{\card{S\cap I_{d/v}} = v}{\card{S\cap I_i} = v~\forall i}
   =\Expect{I_{d/v} : \card{S\cap I_{d/v}} = v}{\cProb{\mathcal{I}'=(I_1,\ldots,I_{d/v-1})}{I_{d/v}}{\card{S\cap I_i} = v~\forall i}},
  \]
  where $\mathcal{I}'$ is a random partition of $[n]\setminus I_{d/v}$ into $d/v-1$ parts. Therefore, by the induction hypothesis it
  is at least $\left(\frac{1}{16\sqrt{v}}\right)^{d/v-1}$, and we are done.
\end{proof}

Equipped with the previous fact, one can deduce the following analog of Corollary~\ref{cor:most_weight_on_good_charachters} using
the same argument used therein.
\begin{corollary}\label{cor:most_weight_on_good_charachters2}
Let $v,d\in\mathbb{N}$, and assume that $v$ divided $d$. Let $f,g\colon\power{n}\to\R$ be functions such that
 $g$ is homogenous of degree $d$. Then there is a partition $\mathcal{I} = (I_1,\ldots,I_{d/v})$ such that
\[
\sum\limits_{\substack{\card{S} = d \\ \card{S\cap I_i} = v~\forall i}}\card{\widehat{f}(S)\widehat{g}(S)}
\geq
\left(\frac{1}{16\sqrt{v}}\right)^{d/v}
\sum\limits_{S\subseteq [n]}\card{\widehat{f}(S)\widehat{g}(S)}.
\]
\end{corollary}

\subsection{Boosted base case}
We can now present a lemma that will be useful for us to make further degree reductions after we have already reached a
degree that is not too large.
\begin{lemma}\label{lem:base_case_improved2}
Let $e_1\leq e_2,\ldots,e_{k-1}\leq e_k$ be given by $e_i = 10^{i}$ and
$\eta_1\geq \eta_2\geq\ldots\geq\eta_{k-1}\geq\eta_k>0$.

Let $f\colon \power{n}\to \power{}$ be a Boolean function and let $g\colon\power{n}\to\R$ be
homogenous of degree $e_k$. Then
\begin{equation}\label{eq:base_case_improved2}
  \inner{f}{g}\leq C_1\left(\sum\limits_{\card{S}=e_k} \card{\widehat{f}(S)}^{\frac{4}{3}}\card{\widehat{g}(S)}^{\frac{4}{3}}\right)^{\frac{3}{4}}
  +C_2\cdot  I[f,g] + C_3\norm{f}_2\norm{g}_2,
\end{equation}
where $C_1, C_2$ and $C_3$ are given as follows.
We have
$C_1 = 10^k \left(16\sqrt{e_{k-1}}\right)^{10} \left(\frac{10}{\eta_k}\right)^{5/2}$.
For $C_2$ and $C_3$ we have that
\[
C_2 = 2\cdot\eta_1^{1/8}3^{e_1/4} \prod\limits_{i=1}^{k-1} (16\sqrt{e_i})^{10},
\qquad\qquad
C_3 = 10^k \prod\limits_{i=1}^{k-1} (16\sqrt{e_i})^{10} \cdot \sum_{j=1}^{k-1} \left(\frac{10}{\eta_j}\right)^{\frac{5}{2}} 3^{e_{j+1}}\eta_{j+1}^{\frac{1}{4}}.
\]
\end{lemma}
\begin{proof}
  The base case $k=1$ is Lemma~\ref{lem:base_case_improved}. Assume the claim holds for $k-1$ and prove it for $k$.

  By flipping signs of the Fourier coefficients of $g$ if necessary,
  we may assume without loss of generality that for all $S$, the signs of $\widehat{f}(S)$ and $\widehat{g}(S)$ are the same.

  Using Corollary~\ref{cor:most_weight_on_good_charachters2}, we may find a partition $\mathcal{I} = (I_1,\ldots,I_{10})$ (recall that $e_k/e_{k-1} = 10$)
  such that for $\tilde{g} = \sum\limits_{\substack{\card{S} = e_k\\\card{S\cap I_i} = e_{k-1}\forall i}}\widehat{g}(S)\chi_S$ we have
  \begin{equation}\label{eq:entropy_1_improved2}
  \inner{f}{g}
  \leq \left(16\sqrt{e_{k-1}}\right)^{10} \inner{f}{\tilde{g}}.
  \end{equation}
  From here, we pretty much repeat the inductive step in the proof of Theorem~\ref{thm:main_formal_fourier_improved}.
  For each $j\in[10]$, and $T\subseteq I_j$, denote $I_{T,I_j}[\tilde{g}] = \sum\limits_{S:S\cap I_j = T}{\widehat{\tilde{g}}(S)^2}$ and similarly for $f$.
  Define $\mathcal{T}_j$ be the set of $T\subseteq I_j$ of size $e_{k-1}$ such that either $I_{T,I_j}[\tilde{g}]\geq \eta_k\norm{\tilde{g}}_2^2$ or
  $I_{T,I_j}[f]\geq \eta_k\norm{f}_2^2$. Then $\card{\mathcal{T}_j}\leq \frac{2}{\eta_k}$.

  Denote by $\mathcal{S}_j$ the set of $S$ of size $e_k$ such that $S\cap I_j\not\in \mathcal{T}_j$, and let
  $\mathcal{B}$ the set of $S$ of size $e_k$ that are outside $\mathcal{T}_1\cup\ldots\cup\mathcal{T}_{10}$. Then we may write that
  \begin{equation}\label{eq:inductive_step_2_improved2}
    \inner{f}{\tilde{g}}\leq
    \underbrace{\sum_{S\in \mathcal{B}}\widehat{f}(S)\widehat{\tilde{g}}(S)}_{(\rom{1})}
    +\underbrace{\sum_{j=1}^{10}\sum_{S\in{\cal S}_j}\widehat{f}(S)\widehat{\tilde{g}}(S)}_{(\rom{2})},
\end{equation}
and we upper bound each sum separately.

\paragraph{Upper bounding $(\rom{1})$.}
Using H\"{o}lder's inequality, $(\rom{1})$ is at most
$\card{\mathcal{B}}^{\frac{1}{4}} \cdot \left(\sum_{\card{S} = e_k} \card{\widehat{f}(S)}^{\frac{4}{3}}\card{\widehat{\tilde{g}}(S)}^{\frac{4}{3}}\right)^{\frac{3}{4}}$.
To bound the size of $\mathcal{B}$, note that the map $S\rightarrow (S\cap I_1,\ldots,S\cap I_{10})$
is a bijection from $\mathcal{B}$ to $\mathcal{T}_1\times\ldots\times\mathcal{T}_{10}$, hence we have
that
\[
\card{\mathcal{B}}^{\frac{1}{4}}\leq
\left(\prod\limits_{i=1}^{10}\card{\mathcal{T}_i}\right)^{\frac{1}{4}}
\leq
\left(\frac{2}{\eta_k}\right)^{\frac{5}{2}}
\]
Therefore, the contribution from
$(\rom{1})$ is upper bounded by the first term on the right hand side in \eqref{eq:base_case_improved2}.
Thus, to complete the proof, it is enough to upper bound the contribution from $(\rom{2})$ by the other
two terms in the right hand side of \eqref{eq:base_case_improved2}, and to do so we use the induction hypothesis.

\paragraph{Upper bounding $(\rom{2})$.}
Writing $g_j = \sum\limits_{S\in\mathcal{S}_j}\widehat{\tilde{g}}(S)\chi_S$, we note that
$(\rom{2}) = \sum\limits_{j=1}^{10}\inner{f}{g_j}$, and it suffices to upper bound each one
of $\inner{f}{g_j}$ separately. Fix $j$ and let $J=\bigcup_{i\neq j} I_i$. Then we may write
\[
\inner{f}{g_j} = \Expect{z\in\power{J}}{\inner{f_{J\rightarrow z}}{(g_j)_{J\rightarrow z}}},
\]
and apply the induction hypothesis for each $z$ separately using the parameters $\eta_1,\ldots,\eta_{k-1}$ and
$e_1,\ldots,e_{k-1}$. We proceed now in the same fashion as in the proof of Theorem~\ref{thm:main_formal_fourier_improved} and
get equation~\eqref{eq:midproof} therein, i.e.\
\[
(\rom{2})\leq
10 \cdot C_1(k-1) 3^{e_k}\eta_k^{1/4}\norm{g}_2\norm{f}_2 + C_2(k-1) I[f,g] + 10\cdot C_3(k-1)\norm{g}_2\norm{f}_2
\]
Note that the first and third term on the right hand side sum up to at most $\frac{C_3(k)}{\left(16\sqrt{e_{k-1}}\right)^{10}}\norm{f}_2\norm{g}_2$.
Thus, plugging this and the bound we have on $(\rom{1})$ into equation~\eqref{eq:inductive_step_2_improved2}, and then into equation~\eqref{eq:entropy_1_improved2},
finishes the proof.
\end{proof}

\begin{corollary}\label{corr:base_case_improved2}
Let $e'$ be a power of $10$, let $\eta>0$ be such that $\eta\leq 10^{-320\sqrt{e'}}$.
Let $f\colon \power{n}\to \power{}$ be a Boolean function and let $g\colon\power{n}\to\R$ be
homogenous of degree $e'$. Then
\[
  \inner{f}{g}\leq
  \frac{1}{\eta^{30 e'}}
  \left(\sum\limits_{\card{S}=e'} \card{\widehat{f}(S)}^{\frac{4}{3}}\card{\widehat{g}(S)}^{\frac{4}{3}}\right)^{\frac{3}{4}}
  +\eta^{1/16}\cdot  I[f,g]
  +\eta\norm{f}_2\norm{g}_2.
\]
\end{corollary}
\begin{proof}
  Denote $e' = 10^k$ and set $e_i = 10^i$ for $i=1,\ldots,k$, $\eta_1=\eta$ and inductively $\eta_{i+1} = \eta_i^{\frac{e_{i+1}}{e_i}\left(1+4\cdot e_{i}^{-1/4}\right)}$.
  We use Lemma~\ref{lem:base_case_improved2} with these parameters, and upper bound each one of $C_1,C_2,C_3$.

  For $C_1$, by choice of $e_{k-1}$ we have that
  $C_1\leq 16^{10}\cdot 10^{6k-5}\left(\frac{10}{\eta_k}\right)^{5/2}\leq 10^{6k+20}\frac{1}{\eta_k^{5/2}}$. By definition, we get that
  \[
  \eta_k = \eta^{e'\prod\limits_{i=1}^{k-1}\left(1+4\cdot e_i^{-1/4}\right)}
  \geq \eta^{e'\prod\limits_{i=1}^{\infty}\left(1+4\cdot 10^{-i/4}\right)}
  \geq \eta^{29e'},
  \]
  and we also note that $10^{6k+20}\leq \eta^{-e'}$ by the upper bound on $\eta$. Hence $C_1\leq \frac{1}{\eta^{30 e'}}$.

  For $C_2$, we have that
  \[
  C_2\leq
  2\eta_1^{1/8} 3^{10/4}\prod\limits_{i=1}^{k-1}10^{5i+20}
  \leq 2\eta^{1/8} 3^{5/2}\cdot 10^{5k^2 + 20k}
  \leq 10^{5k^2+30k} \eta^{1/8}
  \leq \eta^{1/16},
  \]
  where in the last inequality we used the upper bound we have on $\eta$.

  Finally, for $C_3$ we first upper bound each summand in the sum individually.
  Let $1\leq j\leq k-1$, then $\eta_{j+1} = \eta_j^{10\left(1+4\cdot e_j^{-1/4}\right)}$ and so
  \[
  \left(\frac{10}{\eta_j}\right)^{5/2} 3^{e_{j+1}}\eta_{j+1}^{1/4}
  = 10^{5/2}\cdot 3^{e_{j+1}}\cdot \eta_j^{\frac{1}{4}\cdot 10 \cdot 4\cdot e_j^{-1/4}}
  = 10^{5/2}\cdot 3^{10^{j+1}}\eta^{10^{j} \cdot 3\cdot e_j^{-1/4}}
  \leq \eta^2\cdot 3^{10^{j+1}} 10^{-320\sqrt{e'} 10^{3j/4}},
  \]
  where we have used the upper bound we have on $\eta$. Since $e' = 10^k\geq 10^{j+1}$,
  we get that
  \[
  3^{10^{j+1}} 10^{-320\sqrt{e'} 10^{3j/4}}
  \leq 3^{10^{j+1}} 10^{-320\cdot \sqrt{10^{j+1}} 10^{3j/4}} \leq 1.
  \]
  Therefore, each summand in the definition of $C_3$ is at most $\eta^2$, and consequently we get that
  \[
  C_3\leq 10^k \cdot 10^{5k^2+20k}\cdot k\cdot \eta^2 \leq \eta,
  \]
  where in the last inequality we used the upper bound on $\eta$ again.
\end{proof}

\subsection{Boosted general statement}
\begin{thm}\label{thm:main_formal_fourier_improved2}
Let $k\in\mathbb{N}$, $\alpha\in(0,1]$, $\eps\in(0,1/2)$,
let $d_0, e' = d_1,\ldots,d_{k}$ be an increasing sequence such that $d_0=1$ and
for each $2\leq j\leq k$ we have $d_{j-1}\geq \frac{3}{\alpha\eps^2(1-2\eps)^{k}}\log(4(1+\eps)d_{j}/d_{j-1})$,
and let $0<\delta_k\leq\ldots\leq \delta_1 = \eta$. Assume that $\eta\leq 10^{-320\cdot\sqrt{10e'}}$.
Then there are $C_1,C_2, C_3$ specified below, such that the following holds.

If $f\colon \power{n}\to\power{}$ is a Boolean function, and $g\colon\power{n}\to\R$ is $(\alpha,d_k)$-almost-homogenous,
then
\begin{equation}\label{eq:thm_statement_improved}
  \inner{f}{g} \leq C_1\cdot \left(\sum\limits_{\card{S}\sim d_k}\card{\widehat{f}(S)}^{\frac{4}{3}}\card{\widehat{g}(S)}^\frac{4}{3}\right)^{\frac{3}{4}} + C_2\cdot I[f,g]
    +C_3\cdot \norm{g}_2\norm{f}_2.
\end{equation}
For $k = 1$, we have
$C_1 =  10d_1\frac{1}{\eta^{300 e'}}$,
$C_2 = 10d_1\eta^{1/16}$
and $C_3 = 10d_1\eta$.

For $k > 1$, we have
$C_1 = 2^{k} \left(\frac{2}{\delta_k}\right)^{\frac{1+\eps}{4}\frac{d_k}{d_{k-1}}}$,
$C_2 = 2^{k}\cdot 100\cdot d_1^2\cdot \eta^{1/16}$,
and
\[
C_3 =
2^{k}(1+\eps)^k \left(
10d_k \frac{1}{\eta^{300 e'}}\delta_2^{\frac{1}{4}} +
\sum_{j=2}^{k-1} \left(\frac{2}{\delta_j}\right)^{\frac{1+\eps}{4} \frac{d_j}{d_{j-1}}}3^{d_{j+1}}\delta_{j+1}^{\frac{1}{4}}
\right).
\]
\end{thm}
\begin{proof}
  As before, we may assume the signs of Fourier coefficients of $f$ and $g$ match.

  \paragraph{The case $k=1$.}
  This case is essentially given by Corollary~\ref{corr:base_case_improved2}, except that here we are dealing with almost homogenous $g$
  (as opposed to homogenous $g$) and that $e'$ is not necessarily a power of $10$, however both of these issues are very easy to fix.

  Since $g$ is $(\alpha,d_1)$-almost homogenous, it is supported only on characters of size between $\alpha d_1$ and $d_1$,
  and in particular there exists $\alpha d_1\leq d'\leq d_1$ such that for $\tilde{g} = \sum\limits_{\card{S} = d'}\widehat{g}(S)\chi_S$
  we have $\inner{f}{g}\leq d_1\inner{f}{\tilde{g}}$. Next, let $I\subseteq[n]$ be a random subset that contains each $i\in[n]$ independently
  with probability $1/d'$. Then
  \[
  \Expect{I}{\sum\limits_{S: \card{S\cap I}=1}\widehat{f}(S)\widehat{\tilde{g}}(S)}
  =\sum\limits_{\card{S}=d'}\widehat{f}(S)\widehat{\tilde{g}}(S)\Prob{I}{\card{S\cap I}=1}
  \geq\frac{1}{10}\sum\limits_{\card{S}=d'}{\widehat{f}(S)\widehat{\tilde{g}}(S)}
  =\frac{1}{10}\inner{f}{\tilde{g}}.
  \]
  Therefore, defining $\tilde{g}' = \sum\limits_{S:\card{S\cap I}=1}\widehat{\tilde{g}}(S)\chi_S$ we get that
  $\inner{f}{g}\leq d_1\inner{f}{\tilde{g}}\leq 10d_1\inner{f}{\tilde{g}'}$.

  Let $h$ be the closest power of $10$ to $d'$ such that $h\geq d'$, and note that
  $h\leq 10d'\leq 10d_1 = 10e'$.
  We now may replace each variable $x_i$
  for $i\in I$ with $h-d'+1$ variables $x_{i,1},\ldots,x_{h-d'+1}$, substituting each appearances
  of $x_i$ with $x_{i_1}+\ldots+x_{h-d'+1}$. Abusing notations, we think of the functions $\tilde{f},\tilde{g}'$
  as being defined over the new variables, we now note that $\tilde{g}'$ has degree $d'-1 + 1\cdot(h-d'+1) = h$
  (since each character in $\tilde{g}'$ included one variable from $I$ which was replaced by $h-d'+1$ variables).
  We now apply Corollary~\ref{corr:base_case_improved2} on $\tilde{f},\tilde{g}'$, and note that~\eqref{eq:thm_statement_improved}
  readily follows from it. Indeed, for that we note that $\norm{\tilde{f}}_2 = \norm{f}_2$, that
  $\norm{\tilde{g}'}_2\leq \norm{g}_2$ and finally that $I[\tilde{f},\tilde{g}']\leq h\cdot I[f,g]$.

  \paragraph{The inductive step.} The inductive step is performed in the same way as in Theorem~\ref{thm:main_formal_fourier_improved}, and we omit it.
\end{proof}

\begin{corollary}\label{corr:main_set_param_improved2}
  Let $\alpha\in(0,1]$,$d\in\mathbb{N}$ and $\delta>0$ be such that $\delta\leq \frac{1}{d^{96000/\sqrt{\alpha}}}$.

  If $f\colon \power{n}\to\power{}$ is a Boolean function, and $g\colon\power{n}\to\R$ is $(\alpha,d)$-almost-homogenous,
  then
\[
  \inner{f}{g} \leq
  \delta^{-10^{5} d}\cdot \norm{f}_2^{\frac{3}{4}}\norm{g}_2^{\frac{3}{4}}\cdot \max_{\card{S}\sim d}\card{\widehat{f}(S)\widehat{g}(S)}^{\frac{1}{4}}
  +\delta^{1/16}\cdot I[f,g]
  +\delta \cdot \norm{f}_2\norm{g}_2.
\]
\end{corollary}
\begin{proof}
  We intend to use Theorem~\ref{thm:main_formal_fourier_improved2}, and for that we make the choice of parameters.

  Set $d_1 = e' = \frac{1000}{\alpha}(\log d)^2$, and note that by the condition on $\delta$, $\delta\leq 10^{-320\sqrt{10 e'}}$.
  Inductively, set $d_{i+1} = 2 d_i$ until we hit $i=k$ such that $2d_i > d$,
  in which case we set $d_k=d$. Note that clearly, $k\leq \log d$, and denote $\eps = \frac{1}{\log d}$.
  Note that by the choice of parameters, we have that $d_{j-1}\geq \frac{3}{\alpha\eps^2(1-2\eps)^{k}}\log(4(1+\eps)d_{j}/d_{j-1})$ for all $j\geq 2$,
  so the conditions of Theorem~\ref{thm:main_formal_fourier_improved2} on the degrees hold, and we next choose $\eta = \delta_1,\ldots,\delta_k$.

  Choose $\eta = \delta_1 = \delta$, $\delta_2 = \delta_1^{1600d_1}$ and
  inductively $\delta_{i+1} = \delta_i^{2\cdot (1+\eps)^2}$ for $i\geq 2$.
  We note that for $j\leq k-2$ we have $\delta_{j+1} = \delta_2^{(1+\eps)^{2(j-1)} \frac{d_j+1}{d_2}}$,
  and for $j=k-1$ we have $\delta_2^{(1+\eps)^{2(j-1)} 2\frac{d_k}{d_2}}\leq \delta_k\leq\delta_2^{(1+\eps)^{2(j-1)} \frac{d_k}{d_2}}$.

  We now apply Theorem~\ref{thm:main_formal_fourier_improved2} and get that
  \[
  \inner{f}{g}\leq C_1 \left(\sum\limits_{\card{S}\sim d_k}\card{\widehat{f}(S)}^{\frac{4}{3}}\card{\widehat{g}(S)}^{\frac{4}{3}}\right)^{\frac{3}{4}} + C_2\cdot I[f,g] +C_3\cdot \norm{g}_2\norm{f}_2,
  \]
  for $C_1,C_2,C_3$ as in the statement of Theorem~\ref{thm:main_formal_fourier_improved2}. For the first term observe that
  \[
  \sum\limits_{\card{S}\sim d_k}\card{\widehat{f}(S) \widehat{g}(S)}^{\frac{4}{3}}
  \leq \max_{\card{S}\sim d}\card{\widehat{f}(S)\widehat{g}(S)}^{\frac{1}{3}}
  \sum\limits_{\card{S}\sim d}\card{\widehat{f}(S)\widehat{g}(S)}
  \leq \max_{\card{S}\sim d}\card{\widehat{f}(S)\widehat{g}(S)}^{\frac{1}{3}}
  \sqrt{\sum\limits_{S}\widehat{f}(S)^2\sum\limits_{S}\widehat{g}(S)^2},
  \]
  where the last inequality is by Cauchy-Schwarz. By Parseval, this
  is equal to $\max\limits_{\card{S}\sim d}\card{\widehat{f}(S)\widehat{g}(S)}^{\frac{1}{3}}\norm{f}_2\norm{g}_2$.
  Next, we calculate upper bounds for $C_1,C_2,C_3$ for our specific choice of parameters.

  We start with $C_1$. Since $\delta_{k}\geq \delta_2^{(1+\eps)^{2(j-1)} 2\frac{d_k}{d_2}}$ we get that
  \begin{align*}
  C_1
  &= 2^{k} \left(\frac{2}{\delta_k}\right)^{\frac{1+\eps}{4}\frac{d_k}{d_{k-1}}}
  \leq 2^k \cdot 2^{\frac{1+\eps}{4}\frac{d_k}{d_{k-1}}} \cdot \delta_2^{-\frac{1+\eps}{4} (1+\eps)^{2(k-1)}\cdot 2 \frac{d}{d_2}}\\
  &\leq 2^{k+1} \cdot \delta_2^{-10 \frac{d}{d_2}}
  = 2^{k+1}\cdot \delta^{-16000 \frac{d\cdot d_1}{d_2}}
  = 2^{k+1}\cdot \delta^{-16000 d}
  \leq \delta^{-10^5 d},
  \end{align*}
  where in the last inequality we used the upper bound we have on $\delta$ and the fact that $k \leq \log d$.

  Next, for $C_2$ we have
  \[
  C_2 = 2^{k}\cdot 100\cdot d_1^2\cdot \eta^{1/16}
  \leq \delta^{1/32},
  \]
  where we used $k\leq \log d$, $\eta = \delta$ and the upper bound we have on $\delta$.

  Finally, for $C_3$ we upper bound each summand separately.
  Fix $2\leq j\leq k-1$;
  by the definition of $\delta_{j+1}$ we have
  have that the $j$th summand is
  \[
  \left(\frac{2}{\delta_j}\right)^{\frac{1+\eps}{4}\frac{d_{j+1}}{d_j}} 3^{d_{j+1}}\delta_{j+1}^{\frac{1}{4}}
  \leq  \left(\frac{2}{\delta_j}\right)^{\frac{1+\eps}{4} \cdot 2}3^{d_{j+1}}\delta_{j}^{\frac{(1+\eps)^2}{4}\cdot 2}
  \leq 6^{d_{j+1}}\delta_j^{\frac{\eps}{2}}
  \leq 6^{d_{j+1}}\delta^{120\eps d_j}.
  \]
  Since $\delta\leq \frac{1}{d}$ and $\eps = \frac{1}{\log d}$, we get that
  $\delta^{60\eps d_j} \leq 2^{-60 d_j} \leq 2^{-30 d_{j+1}}\leq 6^{-d_{j+1}}$, so that
  we get that the $j$ths summand is at most $\delta^{60\eps d_j}$. Since $d_j\geq d_1\geq \log d$,
  we get that this is at most $\delta^{60}$.

  For $j=1$, the corresponding summand is
  \[
  10d_k \frac{1}{\eta^{300 e'}}\delta_2^{\frac{1}{4}}
  =10d_k \frac{1}{\eta^{300 e'}}\eta^{400e'}
  =10 d_k \delta^{100 e'}
  \leq \delta^{60}
  \]
  by the upper bound on $\delta$.
  Combining, we get that
  \[
  C_3 = 2^k(1+\eps)^k\cdot k\delta^{60}
  \leq \delta.\qedhere
  \]
\end{proof}

\section{Implications of Corollary~\ref{corr:main_set_param_improved2}}\label{sec:imp_corollaries}
In this section, we deduce from Corollary~\ref{corr:main_set_param_improved2} several implications, including
Theorems~\ref{thm:main},~\ref{thm:main_gen} and~\ref{thm:main_entropy_improved} from the introduction.
\subsection{Fourier entropies: proof of Theorem~\ref{thm:main_entropy_improved}}\label{sec:main_ent_improved}
We begin by proving Theorem~\ref{thm:main_entropy_improved}, restated below.
\begin{reptheorem}{thm:main_entropy_improved}
  There exists an absolute constant $K>0$, such that for every $D \in\mathbb{N}$ and any $f\colon\power{n}\to\power{}$ we have that
  \[
     H[\widehat{f^{\leq D}}]\leq K\cdot \sum\limits_{\card{S}\leq D}{|S|\log(|S|+1) \widehat{f}(S)^2} + K\cdot I[f].
  \]
\end{reptheorem}
\begin{proof}
  We first handle the empty character. By the isoperimetric inequality we conclude
  \[
  I[f] \geq {\sf var}(f)\log\left(\frac{1}{{\sf var}(f)}\right)
  =\widehat{f}(\emptyset)\left(1-\widehat{f}(\emptyset)\right)\log\left(\frac{1}{\widehat{f}(\emptyset)\left(1-\widehat{f}(\emptyset)\right)}\right)
  \geq \frac{1}{4}\widehat{f}(\emptyset)^2\log\left(\frac{1}{\widehat{f}(\emptyset)^2}\right),
  \]
  where in the last inequality we used $z(1-z)\log(1/z(1-z))\geq \frac{1}{4} z^2\log(1/z^2)$ that holds for all $z\in[0,1]$.

  For the rest of the proof we assume without loss of generality that $\Prob{x}{f(x) = 1}\leq \half$, since otherwise we may work with the
  function $1-f$. For each $1\leq d\leq D$ and a non-negative integer $k$, define
  \[
    \mathcal{S}_{d,k}
    =\left\{S ~~\middle| \card{S} = d,~
    2^{-C(k+1)\cdot d\log(d+1)} {\sf var}(f)^{1/2} < \card{\widehat{f}(S)}\leq 2^{-C k\cdot d\log(d+1)} {\sf var}(f)^{1/2} \right\},
  \]
  and $g_{d,k} = \sum\limits_{S\in\mathcal{S}_{d,k}}{\widehat{f}(S)\chi_S}$. Our main goal will be to bound the inner
  product $\inner{f}{g_{d,k}}$, which is equal to the Fourier mass of $f$ on $\mathcal{S}_{d,k}$.
  Fix $d\geq 1$ and $k\geq 1$, and set $\delta = \delta(d,k) = 2^{-32k} d^{-10^5}$ and $\alpha = 1$; then by Corollary~\ref{corr:main_set_param_improved2}
  we have, using the upper bound on the Fourier characters, that
  \begin{equation}\label{eq2}
  \inner{f}{g_{d,k}}\leq \delta^{-10^5 d} 2^{-\half C k\cdot d\log(d+1)} \cdot 2{\sf var}(f)  + \delta^{1/16} I[f,g] + \delta \cdot 2{\sf var}(f).
  \end{equation}
  We also used the fact that $\norm{f}_2^2, \norm{g}_2^2\leq 2{\sf var}(f)$.
  Using the definition of $\delta$ and the fact that $C$ is large enough, we get that the first term is at most
  $2^{-\frac{1}{4} C k\cdot d\log(d+1)} {\sf var}(f)\leq 2^{-k}\frac{1}{d^{10}} I[f]$, and the other two terms combined are at most $2^{-k}\frac{1}{d^{10}} I[f]$.
  Therefore, the contribution of $ \mathcal{S}_{d,k}$ to the entropy of $f$ over $d\geq 1$ and $k\geq 1$ is at most
  \begin{align*}
  \sum\limits_{1\leq d\leq D, k\geq 1}\sum\limits_{S\in\mathcal{S}_{d,k}}\widehat{f}(S)^2\log\left(\frac{1}{\widehat{f}(S)^2}\right)
  &\leq \sum\limits_{1\leq d\leq D, k\geq 1}\inner{f}{g_{d,k}}  \left(2C (k+1)\cdot d\log (d+1) + \log\left(\frac{1}{{\sf var}(f)}\right)\right)\\
  &\leq \sum\limits_{d,k\geq 1} 2^{2-k}\frac{I[f]}{d^{10}}  C (k+1)\cdot d\log (d+1) + {\sf var}(f)\log\left(\frac{1}{{\sf var}(f)}\right)\\
  &\leq O(C\cdot I[f]),
  \end{align*}
  where in the last inequality we used the fact that the sum over $k$ and $d$ is $O(1)$, and the isoperimetric inequality.
  Lastly, we upper bound the contribution from $\mathcal{S}_{d,0}$ for $d\geq 1$ to the Fourier entropy of $f$.
  \begin{align*}
  \sum\limits_{1\leq d\leq D}\sum\limits_{S\in\mathcal{S}_{d,0}}\widehat{f}(S)^2\log\left(\frac{1}{\widehat{f}(S)^2}\right)
  &\leq \sum\limits_{1\leq d\leq D}\sum\limits_{S\in\mathcal{S}_{d,0}}\widehat{f}(S)^2 \cdot \left(2C\cdot d\log(d+1) + \log\left(\frac{1}{{\sf var}(f)}\right)\right)\\
  &\leq 2C \sum\limits_{1\leq |S|\leq D}{\card{S}\log(\card{S}+1)\widehat{f}(S)^2} + {\sf var}(f)\log\left(\frac{1}{{\sf var}(f)}\right),
  \end{align*}
  and the second term on the right hand side is upper bounded by $O(I[f])$ using the edge isoperimetric inequality.
\end{proof}

\subsection{Fourier min-entropy: proof of Theorem~\ref{thm:main}}
Next, we note that the argument in the proof of Theorem~\ref{thm:main_entropy_improved} implies Theorem~\ref{thm:main}, restated below.
Recall that $\tilde{I}[f] = \frac{I[f]}{{\sf var}(f)}$.
\begin{reptheorem}{thm:main}
  There exists an absolute constant $C>0$, such that for any $f\colon\power{n}\to\power{}$ there is a non-empty $S\subseteq[n]$ of size
  at most $10 \tilde{I}[f]$, such that $\card{\widehat{f}(S)}\geq 2^{-C\card{S}\log(1+\tilde{I}[f])}\sqrt{{\sf var}(f)}$.
\end{reptheorem}
\begin{proof}
  We run the proof of Theorem~\ref{thm:main_entropy_improved} with $D = 10 \tilde{I}[f]$ with
  slight changes. Namely, we take
    \[
    \mathcal{S}_{d,k}
    =\left\{S ~~\middle| \card{S} = d,~
    2^{-C(k+1)\cdot d\log(1+\tilde{I}[f])} {\sf var}(f)^{1/2} < \card{\widehat{f}(S)}\leq 2^{-Ck\cdot d\log(1+\tilde{I}[f])} {\sf var}(f)^{1/2} \right\},
  \]
  and define $g_{d,k}$ in the same way. Choose $\delta = \delta(d,k) = 2^{-256 k} d^{-10^5}\tilde{I}[f]^{-16}$ and $\alpha = 1$.
  By Corollary~\ref{corr:main_set_param_improved2} for each $k\geq 1$, $d\geq 1$ we have that
  \[
  \inner{f}{g_{d,k}}\leq
  \delta^{-10^5 d} 2^{-\half C k\cdot d\log(1+\tilde{I}[f])} \cdot 2{\sf var}(f)  + \delta^{1/16} I[f,g] + \delta \cdot 2{\sf var}(f).
  \]
  By choice of parameters, the first term is at most $2^{-4k}\frac{1}{d^2} {\sf var}(f)$ given $C$ is sufficiently large constant,
  and it is easy to see that the other two terms are also at most $2^{-4k}\frac{1}{d^2}{\sf var}(f)$. Therefore, we get that
  $\inner{f}{g_{d,k}}\leq 3\cdot 2^{-4k}\frac{1}{d^2} {\sf var}(f)$, and so
  \[
  \sum\limits_{k\geq 1, d\geq 1}\inner{f}{g_{d,k}}
  \leq 3\cdot {\sf var}(f) \cdot\sum\limits_{d\geq 1}{\frac{1}{d^2}} \cdot\sum\limits_{k\geq 1}{2^{-4k}}
  \leq 3\frac{\pi^2}{48} {\sf var}(f)\leq 0.7 {\sf var}(f).
  \]
  Therefore, $\sum\limits_{d\geq 1}\inner{f}{g_{d,0}}\geq {\sf var}(f) - 0.7{\sf var}(f) = 0.3{\sf var}(f)$. Since by Markov's
  inequality we have
  \[
  \sum\limits_{d\geq D}\inner{f}{g_{d,0}} = \sum\limits_{\card{S}\geq D}{\widehat{f}(S)^2}\leq 0.1 {\sf var}(f),
  \]
  it follows that $\sum\limits_{1\leq d\leq D}\inner{f}{g_{d,0}} \geq 0.2 {\sf var}(f) > 0$,
  so there is $d\leq D$ such that $\mathcal{S}_{d,0}\neq \emptyset$, implying there is a character $S\in \mathcal{S}_{d,0}$ as desired.
\end{proof}

\subsection{Fourier concentration: proof of Theorem~\ref{thm:main_gen}}
Next, we show that Theorem~\ref{thm:main_entropy_improved} implies Theorem~\ref{thm:main_gen}, restated below.
\begin{reptheorem}{thm:main_gen}
  For every $\eta>0$, there exists $C>0$, such that for all $f\colon\power{n}\to\power{}$ we have
  \begin{equation}\label{eq:concentrate_improved}
  \sum\limits_{S} \widehat{f}(S)^2 1_{\card{\widehat{f}(S)}\leq 2^{-C\cdot\tilde{I}[f]\log\left(1+\tilde{I}[f]\right)}}
  \leq \eta\cdot {\sf var}(f).
  \end{equation}
\end{reptheorem}
\begin{proof}
By Markov's inequality we have that
$\sum\limits_{\card{S}\geq \frac{2}{\eta}\tilde{I}[f]}{\widehat{f}(S)^2}\leq \frac{\eta}{2}{\sf var}(f)$,
therefore it is enough to bound the contribution of $\card{S}\leq \frac{2}{\eta} \tilde{I}[f]$ to the left hand
side of~\eqref{eq:concentrate_improved} by $\eta/2 \cdot{\sf var}(f)$.
Choose $D = \frac{2}{\eta} \tilde{I}[f]$; by Theorem~\ref{thm:main_entropy_improved}
we get that $H[\widehat{f^{\leq D}}]\leq K (I[f] + I[f]\log D)\leq K'(\eta) I[f]\log(1+\tilde{I}[f])$, thus
\begin{align*}
\sum\limits_{\card{S}\leq D}\widehat{f}(S)^2 1_{\card{\widehat{f}(S)}\leq 2^{-C\cdot \tilde{I}[f]\log\left(1+\tilde{I}[f]\right)}}
&\leq {\sf var}(f)\sum\limits_{\card{S}\leq D}\widehat{f}(S)^2 \frac{\log(1/\widehat{f}(S)^2)}{C \cdot I[f]\log(1+\tilde{I}[f])}\\
&= {\sf var}(f)\frac{H[\widehat{f^{\leq D}}]}{C \cdot  I[f]\log(1+\tilde{I}[f])},
\end{align*}
which is at most ${\sf var}(f) K'/C$. Choosing $C = \frac{2}{K'\eta}$ completes the proof.
\end{proof}

\subsection{An application to transitively symmetric functions}\label{sec:BK_deduce}
In this section, we show how Theorem~\ref{thm:main}
implies an improved form of the Bourgain-Kalai theorem for functions with constant variance.

\skipi

Let $G\subseteq S_n$ be a subgroup. For each subset $S\subseteq[n]$, the
orbit of $S$ under $G$ is ${\sf orbit}_G(S) = \sett{ \pi(S)}{\pi\in G}$.
Define the parameter $a_{T}(G)$ to be the largest $K$, such that each $S$ of size
at most $K$ has $\card{{\sf orbit}_G(S)}\geq 2^{T \card{S} \log(K)}$.

\begin{corollary}\label{corr:BK}
  There are $T>0$ and $c>0$ such that the following holds.
  Let $f\colon\power{n}\to\power{}$ be a function that is symmetric under an invariant subgroup $G\subseteq S_n$. Then
  $I[f]\geq c\cdot a_{T}(G) {\sf var}(f)$.
\end{corollary}
\begin{proof}
  Choose $C > 0$ from Theorem~\ref{thm:main}.
  We prove the statement for $c=\frac{1}{10}$ and $T=2C$.

  Assume towards contradiction that $I[f]< c\cdot a_{T}(G) {\sf var}(f)$; in particular we have
  $\tilde{I}[f]< a_T(G)-1$. By Theorem~\ref{thm:main} there is a non-empty $S\subseteq [n]$
  of size at most $10 \tilde{I}[f]$ such that $\card{\widehat{f}(S)}\geq 2^{-C\card{S}\log(1+\tilde{I}[f])} \sqrt{{\sf var}(f)}$.
  Fix this $S$, and note that by symmetry we have that $\widehat{f}(Q) = \widehat{f}(S)$ for all $Q\in{\sf orbit}_G(S)$, and hence
  \[
  \card{{\sf orbit}_G(S)}\widehat{f}(S)^2
  =\sum\limits_{Q\in{\sf orbit}_G(S)}\widehat{f}^2(Q)
  \leq {\sf var}(f),
  \]
  so $\card{\widehat{f}(S)}\leq \frac{\sqrt{{\sf var}(f)}}{\sqrt{\card{{\sf orbit}_G(S)}}}$. Since $\card{S}\leq 10 \tilde{I}[f]\leq a_T(G)$,
  we get by the definition of $a_{T}(G)$ that $\card{{\sf orbit}_G(S)}\geq 2^{T \card{S}\log a_{T}(G)}$.
  Combining the lower bound and upper bound on $\card{\widehat{f}(S)}$
  we conclude that $2^{\frac{1}{2}T\card{S} \log a_T(G)}\leq 2^{C\cdot\card{S} \log(1+\tilde{I}[f])}$, which since $T = 2C$
  implies that $\tilde{I}[f]\geq a_T(G)-1$, and contradiction.
\end{proof}


\paragraph{Implication for graph properties.} Suppose $f$ is a graph property with constant variance.
In this case, the input is the adjacency vector of length $n = {N\choose 2}$ of an $N$-vertex graph, hence $f$ is symmetric under
the action of permutations on the vertices, i.e.\ of $S_N$. Therefore, the orbit of a collection
of $s\leq N$ edges has size at least ${N \choose \sqrt{s}}\geq \left(\frac{N}{\sqrt{s}}\right)^{\sqrt{s}}\geq 2^{\half \sqrt{s}\log N}$.
Therefore, if we fix $c,T$ from Corollary~\ref{corr:BK}, we see that if $s\leq \left(\frac{\log N}{2T \log\log^2 N}\right)^2$,
then the orbit of a collection of any $s$ edges is at least of size $2^{T s \log(\log^2 N)}$. This implies that
\[
a_T(G)\geq \min\left(\log^2 N, \left(\frac{\log N}{2T \log\log^2 N}\right)^2\right) =\Omega\left(\frac{\log^2 n}{(\log\log n)^2}\right),
\]
hence $I[f] = \Omega((\log n)^2/(\log\log n)^2)$.

\paragraph{Implication for other groups of symmetry.} In general, Corollary~\ref{corr:BK} gives better lower bounds on
$I[f]$ assuming that $f$ is symmetric under some subgroup $G$.

For example, if $G$ is $S_n$ or $A_n$, the orbit of any
$S\subseteq[n]$ such that $\card{S}\leq \sqrt{n}$ has size at least ${n\choose \card{S}/2}\geq n^{c \card{S}} = 2^{c\card{S} \log n}$ for some absolute constant
$c$. Therefore by definition $a_T(G)\geq n^{\Omega(1)}$ for any constant $T$, Corollary~\ref{corr:BK} implies that any $f$ symmetric under $G$ must have at least
polynomially large total influence. Using the result of Bourgain and Kalai, the best bound that could be achieved in this case was of the order $2^{\sqrt{\log n}}$.

\section*{Acknowledgments}
We thank Gil Kalai for valuable conversations over the years and helpful comments. We also thank
Yuval Filmus, Nathan Keller and Subhash Khot for comments about an earlier version of this manuscript.
\bibliographystyle{abbrv}
\bibliography{ref}

\begin{thebibliography}{10}

\bibitem{ACKSW}
S.~Arunachalam, S.~Chakraborty, M.~Kouck\'{y}, N.~Saurabh, and R.~de~Wolf.
\newblock Improved bounds on {Fourier} entropy and min-entropy.
\newblock {\em arXiv preprint arXiv:1809.09819}, 2018.

\bibitem{Beckner}
W.~Beckner.
\newblock Inequalities in {Fourier Analysis}.
\newblock {\em Annals of Mathematics}, 102(1):159--182, 1975.

\bibitem{BenorLinial}
M.~{Ben-Or} and N.~{Linial}.
\newblock Collective coin flipping, robust voting schemes and minima of banzhaf
  values.
\newblock In {\em 26th Annual Symposium on Foundations of Computer Science
  (sfcs 1985)}, pages 408--416, Oct 1985.

\bibitem{Bonami}
A.~Bonami.
\newblock {\'E}tude des coefficients de {Fourier} des fonctions de $l^p(g)$.
\newblock {\em Annales de l'Institut Fourier}, 20(2):335--402, 1970.

\bibitem{Bourgain}
J.~Bourgain.
\newblock On the distribution of the {F}ourier spectrum of {B}oolean functions.
\newblock {\em Israel J. of Math.}, (131):269--276, 2002.

\bibitem{BourgainKalai}
J.~Bourgain and G.~Kalai.
\newblock Influences of variables and threshold intervals under group
  symmetries.
\newblock {\em Geometric and Functional Analysis}, 7(3):438--461, 1997.

\bibitem{CKSS}
S.~Chakraborty, R.~Kulkarni, S.~V. Lokam, and N.~Saurabh.
\newblock Upper bounds on {Fourier} entropy.
\newblock {\em Theoretical Computer Science}, 654:92--112, 2016.

\bibitem{DinurFriedgut}
I.~Dinur and E.~Friedgut.
\newblock Intersecting families are essentially contained in juntas.
\newblock {\em Combinatorics, Probability {\&} Computing}, 18(1-2):107--122,
  2009.

\bibitem{DS}
I.~Dinur and S.~Safra.
\newblock {On the Hardness of Approximating Minimum Vertex Cover}.
\newblock {\em Annals of Mathematics}, 162(1):439--485, 2005.

\bibitem{Friedgut98}
E.~Friedgut.
\newblock {B}oolean functions with low average sensitivity depend on few
  coordinates.
\newblock {\em Combinatorica}, 18(1):27--35, 1998.

\bibitem{FriedgutKalai}
E.~Friedgut and G.~Kalai.
\newblock Every monotone graph property has a sharp threshold.
\newblock {\em Proceedings of the American mathematical Society},
  124(10):2993--3002, 1996.

\bibitem{GL}
O.~Goldreich and L.~A. Levin.
\newblock A hard-core predicate for all one-way functions.
\newblock In {\em Proceedings of the 21st Annual {ACM} Symposium on Theory of
  Computing, May 14-17, 1989, Seattle, Washigton, {USA}}, pages 25--32, 1989.

\bibitem{GKK2}
P.~Gopalan, A.~Kalai, and A.~R. Klivans.
\newblock A query algorithm for agnostically learning dnf?.

\bibitem{GKK}
P.~Gopalan, A.~T. Kalai, and A.~R. Klivans.
\newblock Agnostically learning decision trees.
\newblock In {\em Proceedings of the 40th Annual {ACM} Symposium on Theory of
  Computing, Victoria, British Columbia, Canada, May 17-20, 2008}, pages
  527--536, 2008.

\bibitem{Gross}
L.~Gross.
\newblock {L}ogarithmic {S}obolev {I}nequalities.
\newblock {\em American Journal of Mathematics}, 97(4):1061--1083, 1975.

\bibitem{Hastad}
J.~H{\aa}stad.
\newblock Some optimal inapproximability results.
\newblock {\em J. ACM}, 48(4):798--859, July 2001.

\bibitem{Jackson}
J.~C. Jackson.
\newblock An efficient membership-query algorithm for learning {DNF} with
  respect to the uniform distribution.
\newblock {\em J. Comput. Syst. Sci.}, 55(3):414--440, 1997.

\bibitem{KKL}
J.~Kahn, G.~Kalai, and N.~Linial.
\newblock The influence of variables on {B}oolean functions.
\newblock In {\em {FOCS} 1988}, pages 68--80, 1988.

\bibitem{KalaiPersonal}
G.~Kalai.
\newblock The entropy/influence conjecture.
\newblock
  https://terrytao.wordpress.com/2007/08/16/gil-kalai-the-entropyinfluence-conjecture/,
  2007.
\newblock [Online; accessed 26-October-2019].

\bibitem{KSS}
M.~J. Kearns, R.~E. Schapire, and L.~M. Sellie.
\newblock Toward efficient agnostic learning.
\newblock {\em Machine Learning}, 17(2-3):115--141, 1994.

\bibitem{KellerLifshitz}
N.~Keller and N.~Lifshitz.
\newblock The junta method for hypergraphs and the {Erd\H{o}s-Chv\'{a}tal}
  simplex conjecture, 2017.

\bibitem{KKMO}
S.~Khot, G.~Kindler, E.~Mossel, and R.~O'Donnell.
\newblock Optimal inapproximability results for max-cut and other 2-variable
  csps?
\newblock {\em SIAM J. Comput.}, 37(1):319--357, Apr. 2007.

\bibitem{KhotNaor}
S.~Khot and A.~Naor.
\newblock Nonembeddability theorems via {Fourier} analysis.
\newblock In {\em {FOCS} 2005}, pages 101--112, 2005.

\bibitem{KLW}
A.~R. Klivans, H.~K. Lee, and A.~Wan.
\newblock Mansour's conjecture is true for random {DNF} formulas.
\newblock In {\em {COLT} 2010 - The 23rd Conference on Learning Theory, Haifa,
  Israel, June 27-29, 2010}, pages 368--380, 2010.

\bibitem{RMcapacitiy}
S.~Kudekar, S.~Kumar, M.~Mondelli, H.~D. Pfister, E.~{\c{S}}a{\c{s}}oǧlu, and
  R.~L. Urbanke.
\newblock Reed--muller codes achieve capacity on erasure channels.
\newblock {\em IEEE Transactions on information theory}, 63(7):4298--4316,
  2017.

\bibitem{RMcapacity2}
S.~{Kudekar}, S.~{Kumar}, M.~{Mondelli}, H.~D. {Pfister}, and R.~{Urbankez}.
\newblock Comparing the bit-map and block-map decoding thresholds of
  {R}eed-{M}uller codes on {BMS} channels.
\newblock In {\em 2016 IEEE International Symposium on Information Theory
  (ISIT)}, pages 1755--1759, July 2016.

\bibitem{KM}
E.~Kushilevitz and Y.~Mansour.
\newblock Learning decision trees using the {Fourier} spectrum.
\newblock {\em {SIAM} J. Comput.}, 22(6):1331--1348, 1993.

\bibitem{MansourConj}
Y.~Mansour.
\newblock Learning boolean functions via the fourier transform.
\newblock In {\em Theoretical advances in neural computation and learning},
  pages 391--424. Springer, 1994.

\bibitem{Odonnell}
R.~O'Donnell.
\newblock {\em Analysis of boolean functions}.
\newblock Cambridge University Press, 2014.

\bibitem{OT}
R.~O'Donnell and L.~Tan.
\newblock A composition theorem for the {Fourier} {Entropy-Influence}
  {Conjecture}.
\newblock In {\em {ICALP} 2013}, pages 780--791, 2013.

\bibitem{OWZ}
R.~O'Donnell, J.~Wright, and Y.~Zhou.
\newblock The {Fourier Entropy-Influence Conjecture} for certain classes of
  boolean functions.
\newblock In {\em {ICALP} 2011}, pages 330--341, 2011.

\bibitem{Shalev}
G.~Shalev.
\newblock On the {Fourier Entropy Influence} conjecture for extremal classes.
\newblock {\em arXiv preprint arXiv:1806.03646}, 2018.

\bibitem{WWW}
A.~Wan, J.~Wright, and C.~Wu.
\newblock Decision trees, protocols and the entropy-influence conjecture.
\newblock In {\em ITCS'14}, pages 67--80, 2014.

\end{thebibliography}
\end{document}